\documentclass[10pt,final,journal,letterpaper]{IEEEtran}

\usepackage[]{graphicx}
\usepackage{amssymb, amsmath}
\usepackage{subfig}
\usepackage{algorithm}
\usepackage{algorithmicx}
\usepackage{algpseudocode}

\usepackage{amssymb}
\usepackage{amsmath,amsfonts,amscd}
\usepackage{mciteplus}
\usepackage{cite}
\usepackage{graphicx,epsfig}
\usepackage{citesort}
\usepackage{url}
\usepackage{footmisc}
\usepackage[alwaysadjust]{paralist}
\usepackage{sublabel}
\usepackage{epstopdf}

\newtheorem{proposition}{Proposition}

\newtheorem{definition}{Definition}

\newtheorem{strategy}{Strategy}

\makeatletter

\newcommand{\Rmnum}[1]{\expandafter\@slowromancap\romannumeral #1@}

\newcommand{\ie}{\emph{i.e.}, }

\ifCLASSINFOpdf
\else
\fi
\begin{document}
%
\title{Primary Traffic Characterization and Secondary Transmissions}
%
%
%

\author{Yingxi~Liu,~\IEEEmembership{Student member,~IEEE,}
        Ahmed~Tewfik,~\IEEEmembership{Fellow,~IEEE}\\
\thanks{Y. Liu is with the Department of Electrical and Computer Engineering, University of Texas at Austin, Austin, TX, 78741 USA e-mail: yingxi@utexas.edu.}
\thanks{A. Tewfik is with the Department of Electrical and Computer Engineering, University of Texas at Austin, Austin, TX, 78741 USA e-mail: tewfik@austin.utexas.edu.}}

\markboth{Paper Draft}%
{Shell \MakeLowercase{\textit{et al.}}: Bare Demo of IEEEtran.cls for Journals}
%



\maketitle

\begin{abstract}
Channel idle time distribution based secondary transmission strategies have been studied intensively in the literature. Under various performance metrics, the ultimate performance of secondary devices are eventually dictated by the presumed channel idle time distribution. Such distributions can take any arbitrary form in practice. In this work, we study idle time distributions in wireless local area networks (WLAN) using large amount of the channel idle time data collected in real-world WLAN networks. We demonstrate with experimental data that the channel idle time distribution can be closely modeled by hyper-exponential distribution. Furthermore, one can treat the primary packet arrival process as a semi-Markov modulated Poisson process. Several secondary transmission strategies are discussed under this model. It is shown that using only one hyper-exponential distribution, the secondary user can achieve a desirable performance when the primary packet arrival process is stationary. However, experimental data suggests that in practice, this process is not stationary and the secondary user can experience a large performance loss with stationary transmission strategy. We propose a novel transmission strategy that achieves suboptimal secondary user performance when the idle time distribution is not stationary. The performances of secondary transmission strategies are demonstrated using experimental data.
\end{abstract}

\begin{IEEEkeywords}
cognitive radio, hyper-exponential distribution, Markov-modulated Poisson process
\end{IEEEkeywords}

\IEEEpeerreviewmaketitle

\section{Introduction}
As wireless applications proliferate, the wireless industry is facing a critical problem of decreasing available radio frequencies. It is acknowledged that the concept of cognitive radio \cite{mitola1999cognitive} may offer a solution to this spectrum scarcity problem. Cognitive radio refers to the type of devices that are aware of the spectrum usage in their surroundings and adopt certain mechanism to access the spectrum without affecting other ongoing wireless traffic. The problem of accessing local spectrum without affecting the primary users (PU) has been studied in the frequency, e.g. \cite{ji2007cognitive}, space, e.g. \cite{win2009mathematical}, and time, e.g. \cite{huang2009optimal} domains.

\subsection{Prior work}
A considerable amount of research effort has been spent on the time domain based cognitive radio channel access strategies. In \cite{huang2009optimal, geirhofer2008cognitive}, channel accessing strategies are obtained by assuming that the channel idle time distribution (CITD) follows certain distributions. In \cite{zhao2007decentralized, lai2010cognitive}, the idle and busy states of channels are modeled as Markovian or independently and identically distributed (i.i.d.) over time. The basic structure in most studies is to maximize the access time of the secondary users (SU) subject to a predefined impact on the primary users (PU) performance and based on an assumed CITD. The capacity of the cognitive radio network has been discussion in various aspects. In \cite{4907431, 4373439}, the authors studies the secondary capacity where the primary signal is treated as interference. In \cite{4786456, 5419086}, the secondary capacity is studied in fading channels and the SU needs to meet an interference constraint to the PU. In \cite{4155368, 5208469}, the secondary capacity is studied assuming that the SU has the information of the PU signals. Our study of the secondary capacity differs those prior works in that we focus on the exploration of the instantaneous transmission opportunities rather than transmission power or signal design.

\par
In most of prior works, the assumptions imposed on the CITD are fairly simple. Most of these are inadequate to describe practical primary traffic patterns. In the literature, it is verified that the internet traffic exhibits a heavy-tail and long-range dependence in its traffic (packet) size, transmission time, channel idle time, etc, c.f. \cite{crovella1998heavy}, \cite{paxson1995wide}. This implies that modeling of the CITD with a light tail distribution is inaccurate. Surprisingly, the data collected from several WLAN networks shows that the complementary cumulative distribution function (CCDF) of the CITD features two different characteristics. It follows a power-law decay up to some point. Then it follows an exponential decay. A similar phenomenon is also found in the distribution of inter-connect times between mobile devices \cite{karagiannis2010power}. A careful examination of the WLAN networks reveals that packet arrivals exhibit burst behavior. Markov modulated Poisson process (MMPP) \cite{heffes1986markov} or semi-Markov Modulated Poisson process (SMMPP) are natural tools for modeling this bursty behavior.

\par
The cognitive radio that explores time domain transmission opportunities has a great potential in the future of wireless networks. The technique proposed in this paper is an ad hoc-like secondary network. The primary network can either be ad hoc or centralized controlled networks. It can be applied to many decentralized wireless networks on both licensed and license-free frequency bands. One of the possible application is for the wireless sensor networks to work on the 2.4 GHz or 5 GHz Wi-Fi band. The wireless sensors can exchange information between each other but they do not need to connect to the Wi-Fi network. This technique provides a method for them to connect with decent data rate while remain transparent to the Wi-Fi network. In another scenario, this technique can also relief the primary network from heavy traffic burden. For example, when to hand-held devices want to stream data to each other, they can directly connect to each other rather than connect through an access point. A preliminary version of this technique is named Wi-Fi Direct \cite{alliance2010wi}, which does not consider the protection of other Wi-Fi connection and optimization of its own performance. There also have been other successful techniques, for example, the Bluetooth and infrared ray connection. Nevertheless, they both require additional radio chain and processors and are quite range limited.

\subsection{Contribution}
In this work, the characteristics of the CITD are studied using experimental data. It is found out that the hyper-exponential distribution, which is a mixture of multiple exponential distributions, provides a close fit. Similar approximations of other network metrics can be found in \cite{feldmann1998fitting}. Surprisingly, it turns out that if the durations of the packets are ignored in the SMMPP model, the CITD also follows the hyper-exponential distribution. To study secondary transmission strategies, under the assumption that the primary traffic that it follows the SMMPP model, several secondary transmission strategies are devised to access the channel during the channel idle times. These strategies require different amounts of knowledge of the CITD or equivalently the primary traffic state. Our results shows that the SU only needs to know the hyper-exponential distribution rather than the entire SMMPP to achieve a desirable secondary performance.

\par
The strategies based on the SMMPP model are not robust to model error. A strategy that is robust to errors in the model is designed which tries to access one channel idle time multiple times. This strategy is a refined version of that proposed in \cite{liu2010novel, liuhyperexponential}. To sum up, our contributions in this work are three-fold:
\begin{enumerate}
  \item The hyper-exponential approximation of the CITD is studied and the SMMPP model is proposed which closely models the wireless traffic.
  \item The novel concept of primary traffic state is proposed and a complete study of secondary transmission strategies is provided based on the SMMPP model with consideration of available primary traffic information to the SU.
  \item A novel transmission strategy is proposed when the primary traffic exhibits non-stationarity. The strategy guarantees that the impact on the PU remains consistently under the desired limit.
\end{enumerate}
Furthermore, an experimental study is provide to support the conclusions.

\par
This rest of this paper is organized as follows. Section \ref{sec:2} discusses the hyper-exponential approximation of the empirical CITD. Section \ref{sec:3} briefly discusses the general secondary transmission strategy design problem and provides some definitions. Secondary transmission strategies with primary traffic state information are discussed in Section \ref{sec:4}. Section \ref{sec:5} proposes a strategy for dealing with non-stationary primary traffic. Section \ref{sec:6} provides experimental study results. Section \ref{sec:7} concludes this paper.

\section{Experimental Characterization of Multi-user Multi-service Primary Network Traffic}\label{sec:2}
This section studies the CITD features of a multi-user multi-service primary network. We use the channel activity of WLAN networks as a prototype for packet based random access primary networks in cognitive radio scenario. The channel idle time data is obtained in three locations using USRP2 \cite{kershaw2007kismet}, a software-defined radio device that captures signal with sampling frequency 4 MHz on Wi-Fi channel 11 (center frequency 2.462 GHz). In all of the locations, there are access points (AP) located on a grid, with about 30 meters apart from each other. The first location is in a university cafeteria, with about 5 APs and 40 to 50 moving and static Wi-Fi stations. The second location is in a university office, where there are about 3 to 5 APs and 5 Wi-Fi stations. The third location is in a university library, with about 5 APs and 10 to 20 static Wi-Fi stations. The three data sets, corresponding to the three locations, each contains 100,000 channel idle time samples. The samples are acquired sequentially. Table \ref{tb:dataset1} lists the basic statistics of the data sets.
\begin{table}[!t]
\centering
\caption{Channel idle time statistics observed in three scenarios}
  \begin{tabular}{ l || c | c | c   }
    \hline
     & Set 1 (office) & Set 2 (cafe) & Set 3 (library) \\ \hline \hline
    Set size & 100,000 & 100,000 & 100,000 \\ \hline
    Mean & 0.0012s & 0.0023s & 0.0021s \\ \hline
    Std. Deviation & 0.0023 & 0.0046 & 0.0042 \\
    \hline
  \end{tabular}\label{tb:dataset1}
\end{table}

\subsection{Empirical CITD}
\begin{figure*}[!t]
  \centering
  \subfloat[log-log plot]{\label{fig:loglog}\includegraphics[width=3.5in]{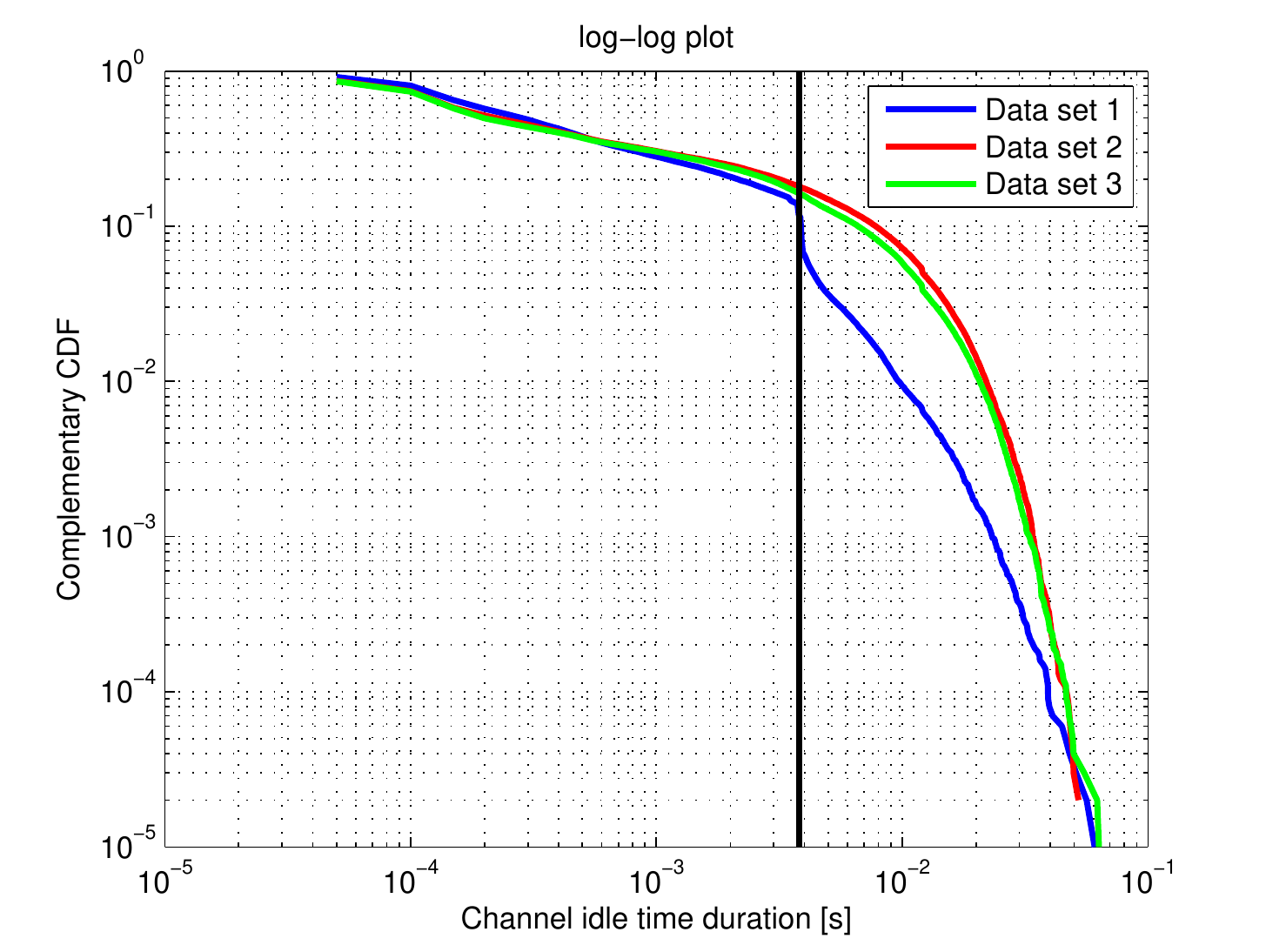}}
  \subfloat[linear-log plot]{\label{fig:linearlog}\includegraphics[width=3.5in]{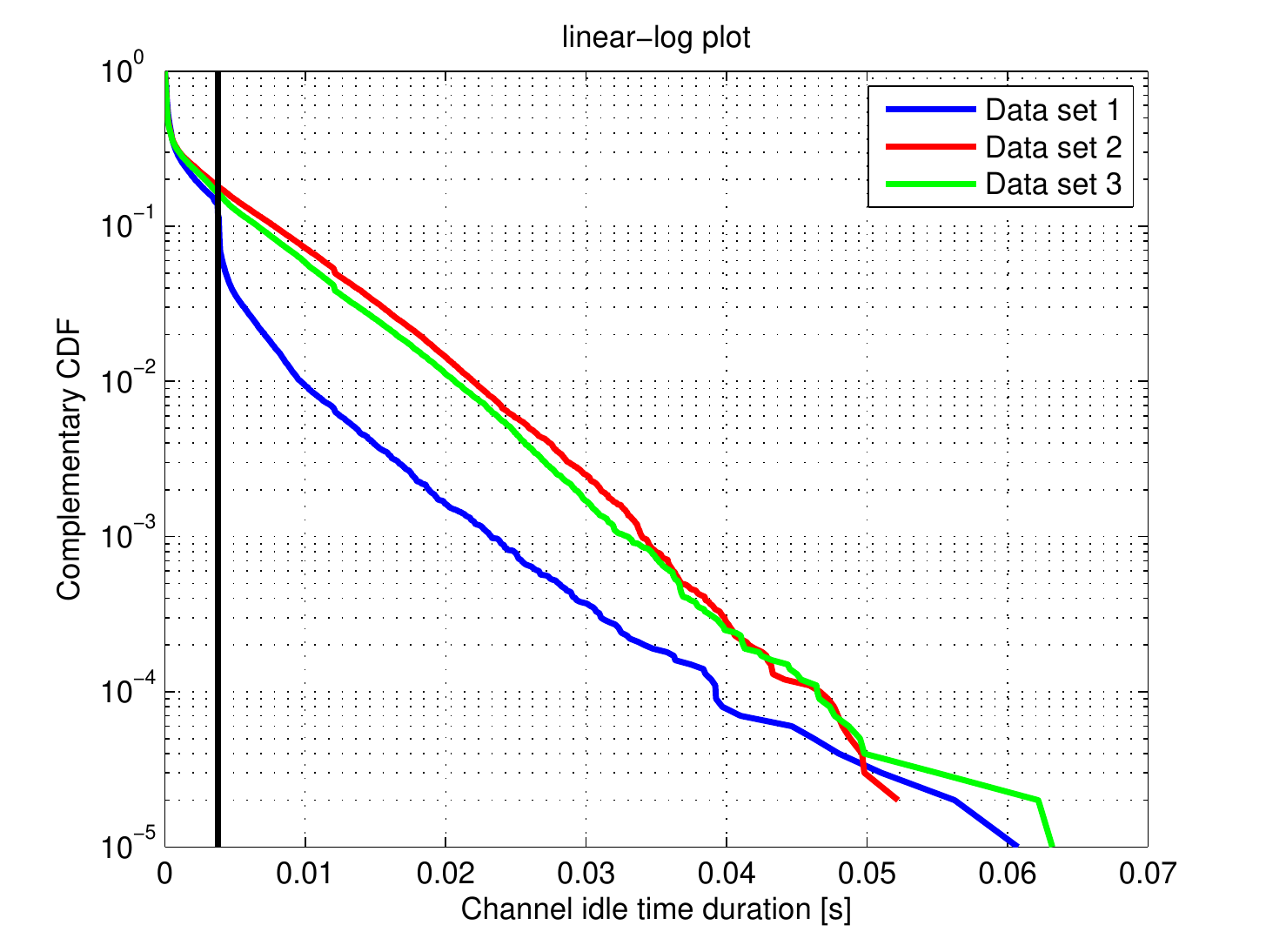}}
  \caption{Log-log and linear-log plots of the CCDF of the data sets}
  \label{fig:logloglinearlog}
\end{figure*}
By carefully examining all the data sets, the evidence is found that the CCDFs of the CITDs follow power-law decay up to some critical point. In Figure \ref{fig:loglog}, the log-log plots are linear up to some point between $10^{-3}$ to $10^{-2}$. The linear-log plot of the CCDFs, as shown in Figure \ref{fig:linearlog}, exhibits linearity beyond the critical point, which confirms that the CITD has an exponential tail.

\par
In the literature, this behavior is modeled as a mixture of exponential distributions, or a hyper-exponential distribution, e.g. \cite{karagiannis2010power}. Additionally, according to Bernstein's theorem \cite{feller2008introduction}, every complete monotone probability distribution function (PDF) is a mixture of exponential PDFs. Although there is no clear evidence showing that the CITD is complete monotone, it can still be approximated as a two-stage PDF which is comprised of a power-law and an exponential decay. To closely approximate a power-law decay usually requires large number of exponentials \cite{feldmann1998fitting}. But a power-law distribution with an exponential decay requires much less exponentials. This feature makes it possible to consider opportunistic secondary access strategies with the hyper-exponential distribution function as a model for primary traffic idle time.

\subsection{Hyper-exponential approximation}\label{sec:3}
The distribution function corresponding to the hyper-exponential distribution is simply the weighted sum of several exponential distributions
\begin{equation}\label{eq:hyperexp}
    f(t)=\sum_{i=1}^{N}\alpha_i\lambda_ie^{-\lambda_it},\ t\geq 0,
\end{equation}
where $N$ is the number of exponentials, $\lambda_i$'s are the corresponding exponential parameters, $\sum_{i=1}^{N}\alpha_i=1$, $\alpha_i\geq 0,\forall i$.

\par
We can estimate the parameters if the hyper-exponential distribution model corresponds to any given traffic pattern using the expectation-maximization (EM) algorithm. Readers are referred to \cite{bilmes1998gentle} for more details about the expectation maximization (EM) algorithm. Two major observations of the EM estimation results are made. Firstly, the number of exponentials needed to approximate the empirical CITD is quite small. In particular, with the data collected, three or four exponentials can produce a reasonably good approximation to the empirical CDF. Increasing the number of exponentials beyond four does not improve the estimate significantly. This is also partly because the tail decays exponentially. Secondly, the estimated $\lambda_i$'s are far apart from each other. For example, the smallest $\lambda$ might be 10 to 100; while the largest $\lambda$ can be 5,000 to 10,000. These two observations will be crucial later for the development of effective secondary transmissions strategies.

\subsection{SMMPP interpretation}\label{sec:3a}
In the semi-Markov modulated Poisson process, the arrival rate is governed by the evolution of a semi-Markov process with discrete state space. Each state is associated with a constant arrival rate. This model is very popular for bursty traffic modeling \cite{heffes1986markov, leland1994self}. Figure \ref{fig:smmpptraffic} shows a simple illustration of a two-state SMMPP. In the first period, the channel is in state 1 where the packets arrive with arrival rate $\lambda_1$. Then it jumps to state 2 with arrival rate $\lambda_2$. Let the steady-state probabilities of the two states be $\alpha_1\geq 0$ and $\alpha_2\geq 0$ with $\alpha_1+\alpha_2=1$. Since the inter-arrival times in the two states follow exponential distributions with exponents of $\lambda_1$ and $\lambda_2$, respectively, one can think of the two channel states to be modulated by the behavior of two users. User 1 in state 1 is video-streaming; hence generates more traffic. User 2 is web browsing and generates less traffic. If the inter-arrival times between state transitions are ignored, the steady-state inter-arrival time distribution follows the weighted sum of the two exponentials $f(t)=\alpha_1e^{-\lambda_1t}+\alpha_2e^{-\lambda_2t}$. Furthermore, if the lengths of the packets are ignored, this distribution would correspond to the distribution of the channel idle times. One can easily extend this two-state semi-Markov process to arbitrary number of states and obtain similar results. Hence, one can see that the SMMPP modeling of the packet arrivals and the hyper-exponential distribution modeling of the CITD are closely related.
\begin{figure}[!h]
\centering
\includegraphics[width=3.5in]{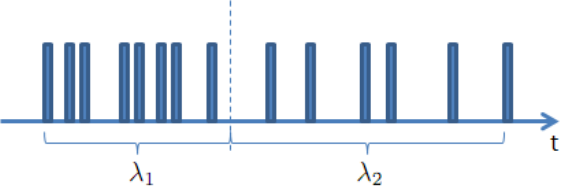}
\caption{A simple example of a two-state Markov modulated Poisson process.}
\label{fig:smmpptraffic}
\end{figure}

\par
In the hyper-exponential distribution, the knowledge of the intensity matrix of the semi-Markov process is not required. On the one hand, knowing the exact state of the semi-Markov process will benefit the design of the secondary transmission strategy. On the other hand, the transmission strategy based on the estimated SMMPP might not be robust to estimation errors. There is not enough evidence showing that the primary traffic follows the SMMPP model with a constant semi-Markov intensity matrix. The main focus is to design the secondary transmission strategies based only on the hyper-exponential distribution.

\section{Secondary Transmission Strategy Preliminaries}\label{sec:3}
This section is devoted to the development of some basic assumptions and concepts related to the secondary transmission strategy.

\par
Firstly, the SU is assumed to be \emph{full-duplex}, \ie the SU can transmit signals and perform spectrum sensing at the same time. This means that if there is a packet collision, the SU can find out immediately after certain detection delay. This seemly unrealistic assumption actually is quite possible given current development in the analog front-end technologies. In \cite{duarte2010full}, it is stated that 40 dB to 80 dB rejection of self-interference is achievable in an analog front-end of a full-duplex radio. Authors in \cite{ahmed2012simultaneous, duarte2012design} improved the self-interference cancellation in the analog front-end. Authors in \cite{choi2010achieving} proposed a novel method which cancels self-interference at the receiving antenna from two transmitting antennas.

\par
Secondly, it is assumed that the SU can perform perfect spectrum sensing when the SU is either idle or transmitting. In practice, the cognitive radio works in slotted time. The spectrum sensing result is produced at the end of each slot. If the slot is sufficiently long or the cognitive radio samples sufficiently fast, a near perfect spectrum sensing result is possible. Nevertheless, the result is not perfect since it is at least a one-slot delayed result. Moreover, spectrum sensing can still have errors. This imperfection can lead to decreased SU performance or increased PU impact. However, the primary channel statistics is not significantly impacted by sensing errors. As a result, the transmission strategy design, which depends only on the channel statistics, remains valid even under imperfect spectrum sensing. The overall performance degradation due to imperfect spectrum sensing is out of the scope of this paper.

\par
The starting time of each channel idle time is denoted as time $0$. The period from the beginning of one channel idle time to the beginning of the next channel idle time is called a \emph{cycle}. The transmission strategy is defined as follows.
\begin{definition}[Secondary transmissions strategy]
At time $t\geq 0$, the SU accesses the channel for an infinitesimal duration with probability $x(t)$ if the channel is sensed to be idle. Otherwise the SU keeps idle.
\end{definition}
The ``infinitesimal duration'' is a mathematical concept. In practice, time can be slotted and the SU makes transmission decision for each slot. A similar definition with slotted time can also be developed. The secondary signal could collide with the primary packets which results in performance loss. The collision situations in slotted time are described next. Due to the assumptions made earlier, there are only two situations that can possibly happen. Either the secondary packet does not overlap with the primary packet (left part of Figure \ref{fig:collision}) leading to a successful secondary transmission, or the secondary packet overlaps with the primary packet in just one slot (right part of Figure \ref{fig:collision}), resulting in a collision. In this situation, it is highly likely that the primary packet cannot be recovered since its packet header is corrupted. The primary packet is \emph{collided} in this case. The probability of collision is defined from the primary users' perspective.
\begin{definition}[Probability of collision]
The probability of collision is defined as the percentage of \emph{collided} primary packets due to secondary transmissions.
\end{definition}
One can immediately see that if $x(t)=1$ $\forall t\geq 0$ when the channel is sensed to be idle, the probability of collision is $1$. To characterize the amount of time that the SU accesses the channel, the concept of secondary capacity associated with a particular secondary transmissions strategy is defined as follows.
\begin{definition}[Secondary capacity]
Secondary capacity is defined as the average SU access time during each channel idle time when a particular secondary transmission strategy is applied.
\end{definition}
Also, if $x(t)=1$ $\forall t\geq 0$ when the channel is sensed to be idle, the secondary capacity is the mean of the channel idle times. Given the CITD $f(t)$, the secondary capacity with strategy $x(t)$ is
\begin{align}
    \int_0^{\infty}x(t)(1-F(t))dt,\nonumber
\end{align}
and the probability of collision is
\begin{align}
    \int_0^{\infty}x(t)f(t)dt.\nonumber
\end{align}
The SU can use different method to design $x(t)$. For example, if the SU knows that the coming or current channel idle time is generated by a certain arrival rate $\lambda_i$, the SU can use a strategy that would boost its secondary capacity. To assist such design, the concept of primary traffic state given the hyper-exponential distribution function \ref{eq:hyperexp} is defined as follows.
\begin{definition}[Primary traffic state]
If the coming or current channel idle time is generated by a certain arrival rate $\lambda_i$, the primary traffic state is $i$.
\end{definition}
There are $N$ states in total.
\begin{figure}[!t]
\centering
\includegraphics[width=3.5in]{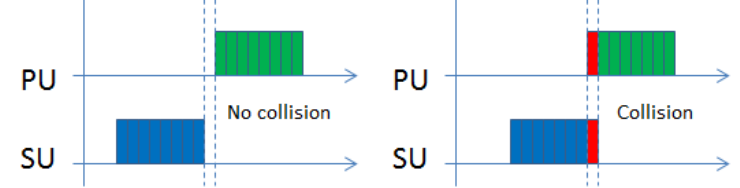}
\caption{Successful secondary transmission and packet collision.}
\label{fig:collision}
\end{figure}

\section{Primary Traffic State Information}\label{sec:4}
This section discusses the performance that the SU can achieve under the assumption that there does exist a SMMPP in the background that controls primary traffic. The SU performance is characterized by the secondary capacity. Depending on how much information the SU has about the exponentials of the CITD at a particular time, secondary capacity varies due to changes in the transmission strategy that the SU uses. This information is referred to as primary traffic state information (PTSI).
\subsection{SMMPP model}
Here, a formal description of the SMMPP model is provided for analysis purposes. Consider $N$ arrival rates $\lambda_1,\lambda_2,\ldots,\lambda_N$. During a particular time period, the channel idle times follow the exponential distribution with a certain $\lambda=\lambda_i,\ i\in\{1,2,\ldots,N\}$. This corresponds to the primary traffic being in state $i$. The transition between states is governed by a recurrent semi-Markov process. To avoid ambiguity, the sojourn time in each state is defined as an integer number of channel idle times. We will denote by $P=\{p_{ij}\}_{i,j=1}^N$ the probability transition matrix, where $p_{ij}$ is the transition probability from state $i$ to state $j$, and $\alpha_i,\ i=1,2,\ldots,N$ the steady-state distribution of the Markov chain.
\subsection{Secondary capacity and PTSI classification}
In the remainder of the paper, we will use the probability of collision to constrain the SU transmission strategy. Then the secondary capacity is a function of the probability of collision $\eta$, which we denote by $T(\eta)$, where $0\leq\eta\leq 1$. We can classify the PTSI into one of three classes, depending on the availability of information about $\lambda_i$'s, $\alpha_i$'s and $P$ for a particular channel idle time. The first class is the statistical PTSI, wherein the SU knows the parameters of the hyper-exponential distribution. The second class is the Markov-level PTSI, wherein the SU knows not only the hyper-exponential distribution but also the Markov transition probabilities in $P$. The third class is the full PTSI, wherein the SU knows the exact arrival rate $\lambda_i$'s. With this class of PTSI, the SU has the most information about the primary traffic. Presumably, the secondary capacity should be the highest with the full PTSI, and the lowest with the statistical PTSI. The three classes are discussed in the sequel.

\subsubsection{Statistical PTSI}
When the SU only knows the parameters $\alpha_i$'s and $\lambda_i$'s, then the SU has \emph{statistical knowledge of the PTSI}. At the beginning of a channel idle time, the SU does not know the current state of the primary traffic. As a result, the SU only knows that the distribution function of the duration of channel idle times as in \eqref{eq:hyperexp}. The statistical PTSI is a very practical assumption for the cognitive radio design. Assuming that the SU has only statistical knowledge of the PTSI, we can develop the following strategies. First, consider a straightforward strategy as follows.
\begin{strategy}[statistical PTSI one-shot strategy]\label{str1}
When a channel idle time comes, the SU transmits from time 0, and stop whenever it detects the primary signal or reaches the time $\tau(\eta)$, which is the maximum duration of transmission determined by the probability of collision constraint:
\begin{align}\label{sec3eq1prime}
    \eta&=\sum_{i=1}^{N}\alpha_i(1-e^{-\lambda_i\tau(\eta)})
\end{align}
and
\begin{align}
    \tau(\eta)\approx\frac{\eta}{\sum_{i=1}^{N}\alpha_i\lambda_i},\nonumber
\end{align}
where the approximation comes from the Taylor expansion for small $\tau(\eta)$.
\end{strategy}
For this straightforward strategy, the secondary capacity is:
\begin{align}\label{sec3eq111}
    T^{\text{stat-PTSI,os}}(\eta)&=\int_{0}^{\tau(\eta)}(1-F(t))dt\nonumber\\
    &=\sum_{i=1}^N\frac{\alpha_i}{\lambda_i}(1-e^{-\lambda_i\tau(\eta)}),
\end{align}
where $F(t)=\int_0^{t}f(t)dt$. We will refer to this strategy as the one-shot strategy. The variable $\tau(\eta)$ is the solution for \eqref{sec3eq1prime}.

\par
Nevertheless, this strategy does not necessarily yield the largest secondary capacity with the statistical PTSI. To explore the longest transmission in a channel idle time given $f(t)$, the SU might need to adopt an intermittent transmission strategy. Let $0\leq x(t)\leq 1$ be the transmission strategy. Consider the probability of collision constrained problem
\begin{align}
    \max\limits_{x(t),t\geq 0}&\ \int_0^{\infty}x(t)(1-F(t))dt,\nonumber\\
    s.t.&\ \int_0^{\infty}x(t)f(t)dt\leq\eta.\nonumber
\end{align}
The optimal solution to this problem yields the largest secondary capacity, which is the following strategy.
\begin{strategy}[statistical PTSI optimal strategy]\label{str2}
The SU first waits from the beginning of each channel idle times until a time $\tau$. If the channel has been idle from the beginning of an idle interval, the SU accesses the channel until it detects a primary packet. We select $\tau$ being such that $1-F(\tau)=\eta$.
\end{strategy}
It is quite noticeable that strategy \ref{str1} and \ref{str2} are in some sense ``opposite'' strategies in that strategy \ref{str1} utilizes the beginning part of channel idle times while strategy \ref{str2} utilizes the tail parts. To see why strategy \ref{str2} results in the largest secondary capacity, one could examine the optimal solution of the probability of collision constrained problem, which is studied in \cite{huang2009optimal}:
\begin{align}\label{sec3statopt}
    x^{\ast}(t)=
    \begin{cases}
    1,\ \frac{1-F(t)}{f(t)}>a^{\ast},\\
    p^{\ast},\ \frac{1-F(t)}{f(t)}=a^{\ast},\\
    0,\ \frac{1-F(t)}{f(t)}<a^{\ast},
    \end{cases}
\end{align}
where $a^{\ast}$ is determined as
\begin{align}
    a^{\ast}=\inf\Big\{a:\int_{t:\frac{1-F(t)}{f(t)}>a}f(t)dt\leq\eta\Big\}.\nonumber
\end{align}
If $\int_{t:\frac{1-F(t)}{f(t)}>a^{\ast}}f(t)dt=\eta$, $p^{\ast}=0$; otherwise
\begin{align}
    p^{\ast}=\frac{\eta-\int_{t:\frac{1-F(t)}{f(t)}>a^{\ast}}f(t)dt}{\int_{t:\frac{1-F(t)}{f(t)}=a^{\ast}}f(t)dt}.\nonumber
\end{align}
This solution is developed for $f(t)$ of a general form. When $f(t)=\sum_{i=1}^{N}\alpha_i\lambda_ie^{-\lambda_it}$, this solution has a special form. It is seen that the transmission opportunities should be allocated to the moments with large $\frac{1-F(t)}{f(t)}$. One can examine that the derivative of $\frac{1-F(t)}{f(t)}$ is nonnegative:
\begin{align}
    \frac{\partial}{\partial t}\frac{1-F(t)}{f(t)}\geq 0.\nonumber
\end{align}
According to the optimal solution \eqref{sec3statopt}, the optimal strategy should allocate transmission opportunity to the times with the largest possible $\frac{1-F(t)}{f(t)}$. Because $\frac{1-F(t)}{f(t)}$ is monotone increasing, the optimal strategy is to transmit in the tail part of channel idle times. Since $1-F(\tau)=\eta$, $\tau$ can be written as a function of $\eta$, $\tau=\tau(\eta)$. Then the secondary capacity with statistical PTSI is indeed
\begin{align}\label{sec3eq3}
    T^{\text{stat-PTSI,opt}}(\eta)=\int_{\tau(\eta)}^{\infty}1-F(t)dt.
\end{align}
One needs to solve a hyper-function to get the value of $\tau(\eta)$ for a particular $\eta$. Therefore there is no close form expression for \eqref{sec3eq3}.

\subsubsection{Markov-level PTSI}
When the SU knows the parameters $\alpha_i$'s, $\lambda_i$'s and the transition probability matrix $P$, the SU has the \emph{Markov-level PTSI}. In theory, with $\alpha_i$'s, $\lambda_i$'s and $P$, the changes from one state to another could be detected. This is known as the Poisson disorder problem \cite{peskir2002solving}. However, to perform tracking of the SMMPP and accessing the idle times at the same time is very challenging, since a change of state might happen while the SU is still detecting the previous state. Basically, it is possible to use this strategy when the primary traffic state changes slowly.

\par
Under this condition, the SU could adopt different transmission strategies since it has the information about the primary traffic state in the previous channel idle time. Given the previous state $i$, the probability distribution function of the upcoming channel idle time is
\begin{align}
    f_i(t)=\sum_{j=1}^{N}p_{ij}\lambda_je^{-\lambda_jt}.\nonumber
\end{align}
Given the knowledge that the primary traffic is in state $i$, similar to the one-shot strategy with statistical PTSI, the SU can choose a transmission strategy according to $f_i$.
\begin{strategy}[Markov-level PTSI one-shot balanced strategy]
When the previous state of primary traffic is $i$, the SU transmits from the start of the next channel idle time until it detects a primary signal or reaches the maximum transmission duration $\tau_i(\eta)$, where
\begin{align}
    \eta&=\sum_{j=1}^{N}p_{ij}(1-e^{-\lambda_j\tau_i(\eta)}).\nonumber
\end{align}
We can solve for $\tau(\eta)$ explicitly as
\begin{align}
    \tau_i(\eta)\approx\frac{\eta}{\sum_{j=1}^{N}p_{ij}\lambda_j},\nonumber
\end{align}
for small $\eta$.
\end{strategy}
This is the one-shot balanced transmission strategy since it requires the same probability of collision for each state $i$. Then the secondary capacity with this strategy is
\begin{align}\label{sec3eq9}
    T^{\text{Markov-PTSI,os-bal}}(\eta)=\sum_{i=1}^N\frac{\alpha_i}{\sum_{j=1}^{N}p_{ij}\lambda_j}\eta.
\end{align}

\par
Besides the balanced strategy, one can associate different probability of collision constraints to different $i$ rather than distribute them equally among all states. This leads to the following optimization problem
\begin{align}\label{sec3prob1}
    \max\limits_{\tau_i\geq 0}&\ \sum_{i=1}^N\alpha_i\sum_{j=1}^N\frac{p_{ij}}{\lambda_j}(1-e^{-\lambda_j\tau_i}),\nonumber\\
    s.t.&\ 1-\sum_{i=1}^N\alpha_i\sum_{j=1}^Np_{ij}e^{-\lambda_j\tau_i}\leq\eta.
\end{align}
However, this problem is not a convex optimization problem since the constraint specifies a non-convex region. One can resort to a suboptimal strategy that allocates the transmission opportunities to the states where the SU could benefit the most.
\begin{strategy}[Markov-level PTSI one-shot suboptimal strategy]
Assume without loss of generality that there exists an order $i_1, i_2, \ldots, i_N$ such that $\sum_{j=1}^N\frac{p_{i_1j}}{\lambda_j}\leq\ldots\leq\sum_{j=1}^N\frac{p_{i_Nj}}{\lambda_j}$ and there exists an $m$ such that $\sum_{k=1}^m\alpha_{i_k}<\eta\leq\sum_{k=1}^{m+1}\alpha_{i_k}$.
\begin{itemize}
  \item When the previous primary traffic state is $i_k$, $1\leq k\leq m$, the SU transmits from the start of each channel idle time until it detects the primary signal.
  \item When the previous primary traffic state is $i_{m+1}$, the SU transmits from the start of each channel idle time until it detects the primary signal or reaches the limit of $\tau_{i_{m+1}}$.
  \item The SU does not transmit at all in all other states.
\end{itemize}
$\tau_{i_{m+1}}$ is determined by the hyper-function
\begin{align}
    \frac{\eta-\sum_{k=1}^m\alpha_{i_k}}{\alpha_{i_{m+1}}}=1-\sum_{j=1}^Np_{i_{m+1}j}e^{-\lambda_j\tau_{i_{m+1}}}.\nonumber
\end{align}
\end{strategy}
The secondary capacity is
\begin{align}\label{sec3eq8}
    T^{\text{Markov-PTSI,os-opt}}(\eta)=\sum_{k=1}^{m}\alpha_{i_{k}}\sum_{j=1}^{N}\frac{p_{{i_{k}}j}}{\lambda_j}+\alpha_{i_{m+1}}\tau_{i_{m+1}}.
\end{align}

\par
The SU can also perform an optimal strategy similar to Strategy \ref{str2}, except that $f(t)$ is replaced with $f_i(t)$. Consider the optimization problem when the previous primary traffic is in state $i$
\begin{align}
    \max\limits_{x_i(t),t\geq 0}&\ \int_0^{\infty}x_i(t)(1-F_i(t))dt,\nonumber\\
    s.t.&\ \int_0^{\infty}x_i(t)f_i(t)dt\leq\eta,\nonumber
\end{align}
Let $x_i^{\ast}(t; \eta)$ be the optimal solution and $T_i(\eta)=\int_0^{\infty}x_i^{\ast}(t; \eta)\big(1-F_i(t)\big)dt$. Here, $x_i^{\ast}(t; \eta)$ is similar to \eqref{sec3statopt}. The strategy is
\begin{strategy}[Markov-level PTSI optimal balanced strategy]
Denote the beginning of the channel idle times as $t=0$. When the previous primary traffic is in state $i$, the SU transmits at time $t$ with probability $x_i^{\ast}(t; \eta)$ for an infinitesimal duration.
\end{strategy}
The secondary capacity is
\begin{align}\label{sec3eq7}
    T^{\text{Markov-PTSI,bal}}(\eta)=\sum_{i=1}^{N}\alpha_iT_i(\eta).
\end{align}

\par
Rather than distributing the probability of collision equally over all states, the SU could be more greedy by exploring the optimal allocation of the probability of collisions over different states. Consider this problem
\begin{align}\label{sec3prob2}
    \max\limits_{x_i(t),\forall i,t\geq 0}&\ \sum_{i=1}^N\alpha_i\int_0^{\infty}x_i(t)(1-F_i(t))dt,\nonumber\\
    s.t.&\ \sum_{i=1}^N\alpha_i\int_0^{\infty}x_i(t)f_i(t)dt\leq\eta.
\end{align}
The solution to this problem could be found out more clearly through the slotted version of it. Let $p_{ij}=F_i(j\Delta)-F_i((j-1)\Delta)$ and $Q_{ij}=1-F_i((j-1)\Delta)$ for $j=1,2,\ldots,K$ and $i=1,\ldots, N$. The slotted version of this problem is
\begin{align}
    \max\limits_{x_i(j),\forall i, \forall j}&\ \sum_{i=1}^N\alpha_i\sum_{j=1}^{K}Q_{ij}\Delta x_i(j),\nonumber\\
    s.t.&\ \sum_{i=1}^N\alpha_i\sum_{j=1}^{K}p_{ij}x_i(j)\leq\eta.\nonumber
\end{align}
The optimal choice is to allocate the transmission opportunities $x_i(j)$ to the slots with the largest $\frac{Q_{ij}}{p_{ij}}$ possible, for all $i=1,\ldots,N$ and $j=1,\ldots,K$. This is essentially the same problem as the slotted strategy in the statistical PTSI. In the continuous case, the optimal solution is as the following.
\begin{proposition}\label{sec3prop1}
The optimal solution of \eqref{sec3prob2} is
\begin{align}\label{sec3markovopt}
    x_i^{\ast}(t)=
    \begin{cases}
    1,\ \frac{1-F_i(t)}{f_i(t)}>a^{\ast},\\
    p^{\ast},\ \frac{1-F_i(t)}{f_i(t)}=a^{\ast},\\
    0,\ \frac{1-F_i(t)}{f_i(t)}<a^{\ast},
    \end{cases}
\end{align}
for all $i$. Here, $a^{\ast}$ is such that
\begin{align}
    a^{\ast}=\inf\Big\{a:\sum_{i=1}^N\alpha_i\int_{t:\frac{1-F_i(t)}{f_i(t)}>a}f_i(t)dt\leq\eta\Big\},\nonumber
\end{align}
and $p^{\ast}$ is
\begin{align}
    p^{\ast}=\frac{\eta-\sum_{i=1}^N\alpha_i\int_{t:\frac{1-F_i(t)}{f_i(t)}>a}f_i(t)dt}{\sum_{i=1}^N\alpha_i\int_{t:\frac{1-F_i(t)}{f_i(t)}=a^{\ast}}f_i(t)dt}.\nonumber
\end{align}
\end{proposition}
Analogous to the slotted case, the SU transmits in the time that has the largest $\frac{1-F_i(t)}{f_i(t)}$ possible when the previous primary traffic state is $i$. Proof of this proposition is in Appendix \ref{app2}. The strategy with optimal allocation over time and states can be summarized as
\begin{strategy}[Markov-level PTSI optimal strategy]
Let the beginning of each idle time be $t=0$. When the previous primary traffic state is $i$, the SU transmits for an infinitesimal duration at time $t$ with probability $x_i^{\ast}(t)$.
\end{strategy}
The secondary capacity with this strategy is
\begin{align}
    T^{\text{Markov-PTSI,opt}}(\eta)=\sum_{i=1}^N\alpha_i\int_0^{\infty}x_i^{\ast}(t)(1-F_i(t))dt.
\end{align}

\subsubsection{Full PTSI}
When the SU knows the state of the primary traffic at the beginning of each channel idle time, then the SU has \emph{full PTSI}. With the full PTSI, the SU knows for sure that which exponential distribution generates the upcoming channel idle time. This assumption is practical under the condition that the primary traffic state changes very slowly. However, with full PTSI, the SU can utilize the channel idle times in the most efficient way, which provides an upper bound on the secondary capacity using opportunistic accessing strategies.

\par
Firstly, the SU can adopt a balanced strategy.
\begin{strategy}[full PTSI balanced strategy]
The SU starts to transmit at the beginning of each channel idle time until it detects the primary signal or it reaches the transmission duration limit $\tau_i$ which is determined by the probability of collision
\begin{align}
    1-e^{-\lambda_i\tau_i}=\eta.\nonumber
\end{align}
This implies that the average SU transmission duration is given by
\begin{align}
    \tau_i(\eta)&=\frac{1}{\lambda_i}\log\frac{1}{1-\eta}\nonumber\\
    &=\frac{1}{\lambda_i}(\frac{\eta}{1-\eta}+o(\eta^2))\nonumber\\
    &\approx\frac{\eta}{\lambda_i}.\nonumber
\end{align}
The approximation comes from Taylor expansion with small $\eta$.
\end{strategy}
The secondary capacity under this strategy is
\begin{align}\label{sec3eq11}
    T^{\text{full-PTSI,bal}}(\eta)=\sum_{i=1}^N\frac{\alpha_i}{\lambda_i}\eta.
\end{align}

\par
Similarly, the SU could distribute the probability of collision optimally over all states. This results in an optimization problem similar to \eqref{sec3prob2}, except that $f_i(t)$ is an exponential distribution function instead of a hyper-exponential distribution function. According to the results in proposition \ref{sec3prop1}, the SU should allocate the transmission opportunities to the states with the largest possible $\frac{1-F_i(t)}{f_i(t)}=\frac{1}{\lambda_i}$. The optimal allocation strategy is as follows.
\begin{strategy}[full PTSI optimal strategy]
Assume without loss of generality that $\lambda_1\leq\ldots\leq\lambda_N$, and there exists an $m$ such that $\sum_{i=1}^m\alpha_i<\eta\leq\sum_{i=1}^{m+1}\alpha_i$.
\begin{itemize}
  \item When the primary traffic is in state $i\leq m$, the SU transmits until it detects the primary signal.
  \item When the primary traffic is in state $m+1$, the SU transmits until it detects the primary signal or reaches the transmission duration limit $\tau_{m+1}=\frac{1}{\lambda_{m+1}}\log\frac{1}{1-\eta}$.
  \item In all other states, the SU does not transmit at all.
\end{itemize}
\end{strategy}
The secondary capacity is
\begin{align}\label{sec3eq10}
    T^{\text{full-PTSI,opt}}(\eta)=\sum_{i=1}^{j}\frac{\alpha_i}{\lambda_i}+\frac{1}{\lambda_{j+1}}(\eta-\sum_{i=1}^j\alpha_i).
\end{align}

\subsection{Comparison of the three types of PTSI}
Among the three types of PTSI, it is most convenient for the SU to have the statistical PTSI. The Markov-level PTSI requires the SU not only to have the statistical PTSI, but also the transition probabilities of the primary traffic states, which could be estimated from the channel idle time data. However, if the primary traffic is fast-changing, the SU could not correctly track the states. Then knowing the transition probabilities would not help improving the secondary capacity. When the primary traffic state changes slowly, the SU needs to perform $N-1$ likelihood ratio tests for every channel idle time. This requirement is challenging but not impossible. The difference between the full PTSI and the Markov-level PTSI is that the full PTSI means the prior information of the primary traffic state of the upcoming channel idle time is known. This requirement is clearly unrealistic in the fast-changing case, but it can be used in the slow-changing case. The secondary capacity with full PTSI provides a reasonable upper bound for strategies with statistical or Markov-level PTSI.

More knowledge about the primary traffic usually results in more transmission opportunities. This is easily recognized by looking at the one-shot balanced strategies with the three types of PTSI. Because of the convexity of the function $f(t)=\frac{1}{t}$, it is obvious that $T^{\text{full-PTSI,bal}}(\eta)\geq T^{\text{stat-PTSI,os}}(\eta)$ because of Jensen's inequality. With slow-changing primary traffic state, the transition probability $p_{ij}$ is small for $j\neq i$. Then the approximation $\sum_{k=1}^{N}p_{ik}\lambda_k\approx\lambda_i$ is valid. Hence, we have $T^{\text{Markov-PTSI,os-bal}}(\eta)\approx T^{\text{full-PTSI,bal}}(\eta)$.

\par
Compared with the one-shot strategies, the gain using the optimal strategy can be large. In the slotted optimal transmission strategy with statistical PTSI, the quantity $\frac{Q_i}{p_i}$ has the meaning of the ``value-to-cost'' ratio, \ie the ratio between the secondary capacity and the probability of collision when the SU transmits in slot $i$. The slotted optimal transmission strategy suggests that the SU should always try to use the slots with the largest value-to-cost ratio. When the slot length $\Delta$ goes to zero we get
\begin{align}
    \lim\limits_{\Delta\rightarrow 0}\frac{Q_i\Delta}{p_i}&=\lim\limits_{\Delta\rightarrow 0}\frac{Q_i}{\frac{p_i}{\Delta}}\nonumber\\
    &=\frac{1-F(t)}{f(t)}.\nonumber
\end{align}
This is consistent with the optimal strategy in the continuous case which requires that the SU should only transmit at the times that $\frac{1-F(t)}{f(t)}$ is above a certain level. One can write explicitly that
\begin{align}
    \frac{1-F(t)}{f(t)}=\frac{\sum_{i=1}^{N}\alpha_ie^{-\lambda_it}}{\sum_{i=1}^{N}\alpha_i\lambda_ie^{-\lambda_it}}.
\end{align}
When $t\rightarrow 0$, $\frac{1-F(t)}{f(t)}\approx\frac{1}{\sum_{i=1}^{N}\alpha_i\lambda_i}$. When $t$ is large enough, $\frac{1-F(t)}{f(t)}\approx\frac{1}{\min\limits_{i}\lambda_i}$. So the difference between the one-shot strategies and the optimal strategy is approximately the difference between $\frac{1}{\sum_{i=1}^{N}\alpha_i\lambda_i}$ and $\frac{1}{\min\limits_{i}\lambda_i}$, which can lead to a big gap implied by the result in the previous section.

\par
Without analytical forms of the secondary capacity with the optimal transmission strategies, it is not straightforward to compare the performances in the three types of PTSI. A simulation study is provided in the sequel. Consider three traffic levels $\lambda_1=5$, $\lambda_2=100$ and $\lambda_3=6000$. The transition probability matrix
\begin{align}
    P=\begin{bmatrix}
        0.9 & 0.05 & 0.05 \\
        0.05 & 0.9 & 0.05 \\
        0.05 & 0.05 & 0.9 \\
      \end{bmatrix},\nonumber
\end{align}
and $\alpha_1=\alpha_2=\alpha_3=\frac{1}{3}$. The secondary capacities of the 8 strategies are plotted in Figure \ref{fig:sc}. Several observations are made through the secondary capacities. Firstly, the SU could always get a good secondary capacity by optimally allocating the transmission opportunities over time. The solid blue, green, red and black curves are the optimal strategies that yields the most secondary capacities. The optimal strategy with full PTSI yields the most secondary capacity. The optimal strategies with the statistical and Markov-level PTSI have similar secondary capacities, but both slightly lower than $T^{\text{full-PTSI,opt}}$. The optimal balanced strategy with the Markov-level PTSI is again a little lower, and in some cases, it is lower than $T^{\text{Markov-PTSI,os-opt}}$. Secondly, the optimal allocation over the primary traffic states with the Markov-level PTSI can bring a boost in the secondary capacity. One can compare $T^{\text{Markov-PTSI,opt}}$, $T^{\text{Markov-PTSI,bal}}$ and $T^{\text{Markov-PTSI,os-opt}}$, $T^{\text{Markov-PTSI,os-bal}}$. Thirdly, having more information does not necessarily mean higher secondary capacity. With the statistical PTSI, the SU can still get a performance close to the situation with the full PTSI when the right strategy is chosen.
\begin{figure}[!t]
\centering
\includegraphics[width=3.5in]{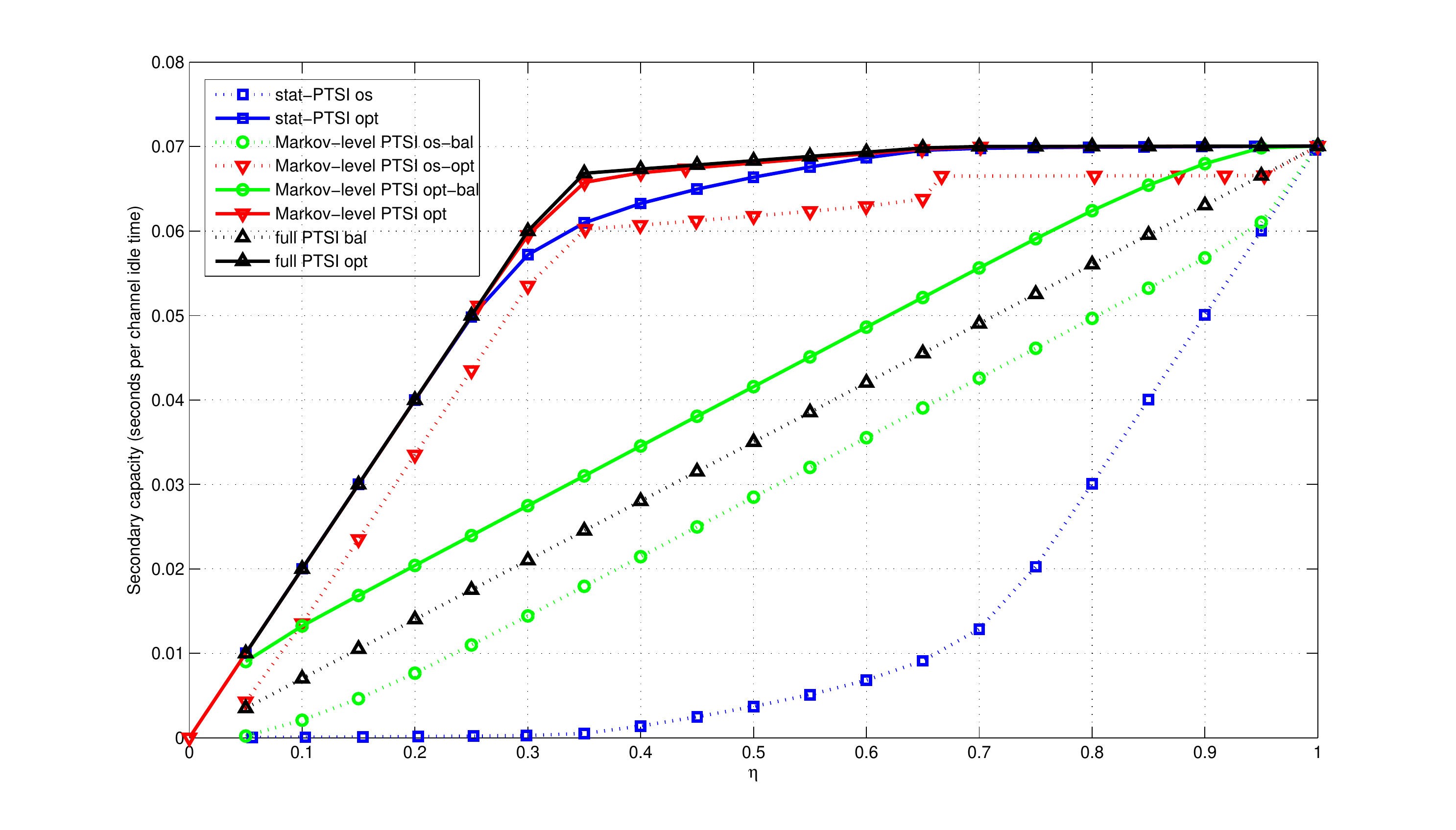}
\caption{Secondary capacities of the 8 strategies.}
\label{fig:sc}
\end{figure}

\par
In a practical situation, the SU can conveniently obtain the statistical PTSI. The full PTSI is not considered as a practical prerequisite. Although Markov chains could be estimated, it still puts an enormous burden on the SU. The optimal strategy with the statistical PTSI could reach a performance close to the optimal. One question then is, ``does there exist a strategy with the statistical PTSI that has simple solutions and is also robust to errors in $\alpha_i$'s and $\lambda_i$'s?'' The next section provides an affirmative answer.

\section{Secondary Transmission Under Non-stationary Primary Traffic}\label{sec:5}
The strategies with the statistical and Markov-level PTSI require some level of stationarity in the primary traffic. The stationarity in the statistical PTSI requires that the $\alpha_i$'s and $\lambda_i$'s remain static. The stationarity in the Markov-level PTSI requires not only the $\alpha_i$'s and $\lambda_i$'s remain static, but also the transition probabilities have to stay static. In practice, it is usually difficult to have static transition probability matrix. As a result, it is impractical to consider the Markov-level PTSI. In most situations, it is even difficult to have static $\alpha_i$'s and $\lambda_i$'s. In such situations, the performance of the strategies with the statistical PTSI will be affected in the sense that its probability of collision will significantly deviate from the designed value. If the primary network has a limit on the largest probability of collision it can tolerate, it is desirable that the SU keeps the probability of collision under the designed limit.

To overcome this drawback and maintain a stable probability of collision performance, a strategy that depends only on the $\lambda_i$'s but not the $\alpha_i$'s is devised. The variation of the probability of collision can be reduced since it can only be affected by changes in the $\alpha_i$'s. At the same time, it also reaches a secondary capacity that is comparable to the statistical PTSI optimal strategy.

\subsection{Multiple-shot Strategy}
The basic idea of the multiple-shot strategy is as follows. Assume without loss of generality that $\lambda_1\leq\lambda_2\leq\ldots\leq\lambda_N$. At the beginning of each channel idle time, which is denoted as time $t=0$, the only prior information the SU has is the coefficients of the hyper-exponential distribution function \eqref{eq:hyperexp}, i.e., the $\alpha_i$'s and $\lambda_i$'s. Instead of following the one-shot strategy, to be conservative, the SU first assume that this idle time follows the fastest exponential decay, i.e., exponential distribution with $\lambda_N$. Given such a guess, one can compute the maximum duration of secondary transmission $t_N$ such that the probability of collision is preserved
\begin{align}
    1-e^{-\lambda_Nt_N}=\eta,\nonumber
\end{align}
which yields
\begin{align}
    t_N&=\frac{1}{\lambda_N}\log\frac{1}{1-\eta}\nonumber\\
    &\approx\frac{\eta}{\lambda_N}.\nonumber
\end{align}
After this transmission, the SU continues to observe the channel. If the channel remains idle from time $t_N$ to time $t_N^e$, the SU is almost sure that this idle time is not generated by the exponential distribution with $\lambda_N$, where $t_N^e$ is determined by
\begin{align}
    1-e^{-\lambda_Nt_N^e}=1-\epsilon,\nonumber
\end{align}
where $\epsilon>0$ is a very small number. This implies that
\begin{align}
    t_N^e=\frac{1}{\lambda_N}\log\frac{1}{\epsilon}.\nonumber
\end{align}
Then the SU knows that with probability $1-\epsilon$ this idle time is not generated by $\lambda_N$. The SU then assumes that the idle time is generated by $\lambda_{N-1}$. Given this conservative guess as well as the fact that this idle time is longer than or equal to $t_N^e$, the SU can access the channel for another duration $t_{N-1}$ starting from $t_N^e$, which is determined as follows
\begin{align}
    \eta&=P(t_N^e\leq X< t_N^e+t_{N-1}|X\geq t_N^e, \lambda_{N-1})\nonumber \\
    &=1-e^{-\lambda_{N-1}t_{N-1}},\nonumber
\end{align}
where $X$ is the random variable of the channel idle time. This yields
\begin{align}
    t_{N-1}\approx\frac{\eta}{\lambda_{N-1}}.\nonumber
\end{align}
Again, if the channel remains idle from time $t_N^e+t_{N-1}$ to some time $t_{N-1}^e$ such that $1-e^{-\lambda_{N-1}t_{N-1}^e}=1-\epsilon$, the secondary device can transmit for a duration of $t_{N-2}$ starting from $t_{N-1}^e$ by assuming that the idle time follows the exponential distribution with $\lambda_{N-2}$ and given that the idle time is longer than $t_{N-1}^e$. $t_{N-2}$ can be derived similarly. Figure \ref{fig:ms} shows an example of the above process.
\begin{figure}[!t]
\centering
\includegraphics[width=3.5in]{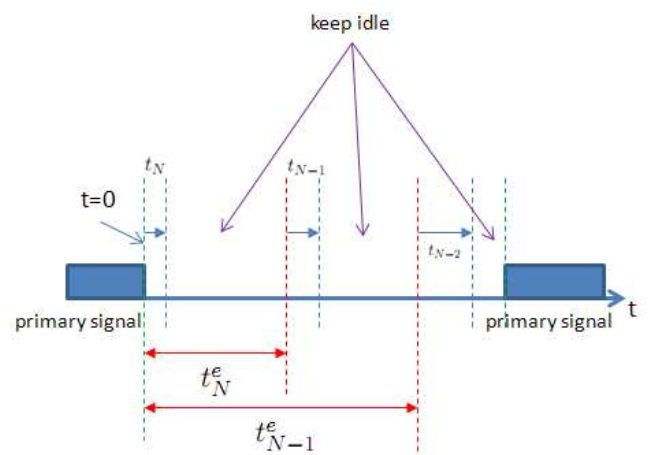}
\caption{An example of multiple-shot transmission strategy.}
\label{fig:ms}
\end{figure}

\par
The multiple-shot strategy is summarized as follows.
\begin{strategy}[statistical PTSI multiple-shot strategy]
The strategy follows three steps.
\begin{enumerate}
  \item The SU starts transmitting from the beginning of each channel idle time $t=0$ until it detects the primary signal or reaches the transmission duration limit $t_N$, which is determined by the probability of collision constraint $1-e^{-\lambda_Nt_N}=\eta$. If the channel remains idle during the time $[0, t_N]$, goes to the next step.
  \item If the channel remains idle till the time $t_{N-j}^e$ such that $1-e^{-\lambda_{i}t_{N-j}^e}= 1-\epsilon$, it transmits from the time $t_{N-j}^e$ until it detects the primary signal or reaches the transmission duration limit  $t_{N-j}$ which satisfies $1-e^{-\lambda_Nt_{N-j}}=\eta$ for $j=0, 1, 2, \ldots, N-1$. If the channel remains idle during the transmission, repeat this step.
  \item If primary signal appears and the channel turns busy, secondary user enters observing state and wait for the channel to be idle. Go to step 1.
\end{enumerate}
\end{strategy}

\par
An important feature of this multiple-shot strategy is that under this strategy, the probability of collision $\eta$ is preserved under the $\lambda_i\gg \lambda_{i-1}$ condition. This is illustrated in the following proposition.
\begin{proposition}\label{sec4prop1}
The probability of collision under the multiple-shot strategy is less than $\eta$ given that the CITD follows the associated hyper-exponential distribution.
\end{proposition}
\begin{proof}
See Appendix \ref{app2}.
\end{proof}

\par
According to this proposition, the probability of collision of the multiple-shot strategy is at least as good as the one-shot strategy. The average of the secondary access time of the multiple-shot strategy is
\begin{align}\label{sec4eq2}
    T^{\text{stat-PTSI,ms}}(\eta)=\sum_{j=1}^{N-1}\alpha_j\sum_{i=j}^{N-1}\frac{e^{-\lambda_jt_{i+1}^e}\eta}{\lambda_{i}}+\frac{\eta}{\lambda_N}.
\end{align}
One should notice that this equation is valid for small $\eta$.

\par
The benefit of the multiple-shot strategy over the one-shot strategy can be explained as follows. In the one-shot strategy, it is automatically assumed that the CITD follows the fastest exponential decay rate and result in a very short transmission time $t_N$. In the multiple-shot strategy, the probability of collision is spread out in each step. Whenever an idle time follows a slow decay rate, the multiple-shot strategy will benefit from it by generating a longer access time. This can also be understood from \eqref{sec4eq2} where each term $\frac{1}{\lambda_i}$ is linearly combined rather than nonlinearly as in \eqref{sec3eq111}. Therefore, the benefit of small $\lambda_i$'s would not be penalized by large $\lambda_i$'s. The secondary capacity for $\eta\in[0,0.1]$ is plotted in Figure \ref{fig:sc1}. Its secondary capacity is suboptimal compared to the optimal strategies. But it is very close to that of the full PTSI balanced strategy. Moreover, the strategy design is completely independent of the $\alpha_i$'s. As a result, this strategy is insensitive to estimation errors in $\alpha_i$'s. This feature will be illustrated in detail in the next section.
\begin{figure}[!t]
\centering
\includegraphics[width=3.5in]{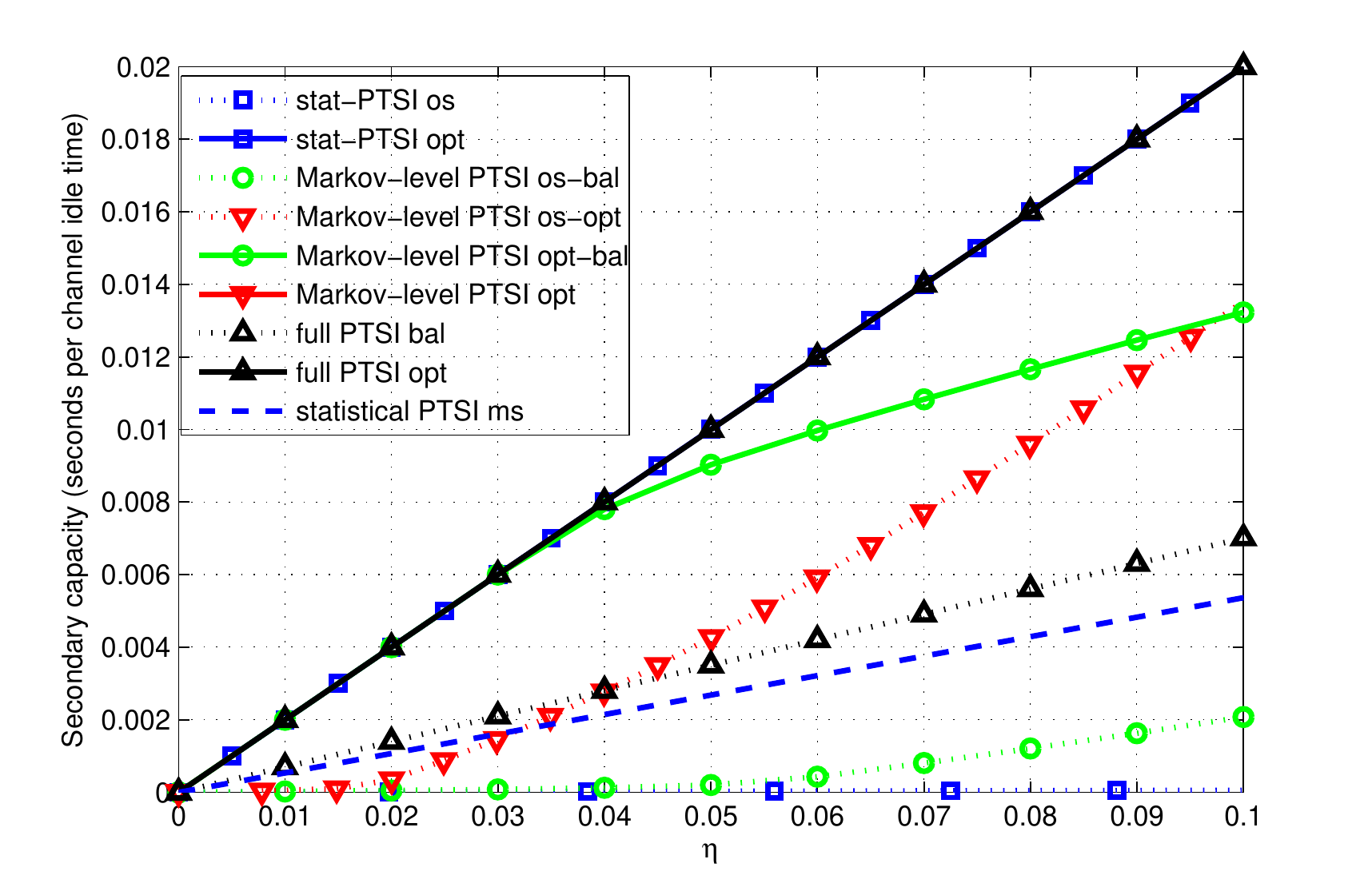}
\caption{Secondary capacities of multiple-shot strategy and the other 8 strategies.}
\label{fig:sc1}
\end{figure}

\subsection{Probability of collision in non-stationary primary traffic}
Because the value-to-cost ratio of hyper-exponential distribution is non-decreasing in $t$, the statistical PTSI optimal strategy often ends up with a strategy that forces the SU to keep idle from the beginning of a channel idle time for a specified period, and then transmit until it detects a primary packet. Therefore, even if the primary traffic has a different hyper-exponential distribution, this strategy is still optimal in the sense that for the specific probability of collision it generates, it achieves the highest secondary capacity. Nonetheless, to control the probability of collision.

\par
The following example is considered to show how the probability of collision can change in non-stationary primary traffic with statistical PTSI optimal and multiple-shot strategies. Consider hyper-exponential distribution with $\lambda_1=100$, $\lambda_2=6000$ with weights of $\alpha$ and $1-\alpha$. The SU treat $\alpha=0.5$. But the true value is $\alpha=0.5,0.6,0.7,0.8$ or $0.9$. The designed probability of collision is $0.1$. Figure \ref{fig:pcbar} plots the actual probability of collision with these two strategies. One can see that the multiple-shot strategy preserves the probability of collision while the optimal strategy loses control of it.
\begin{figure}[!h]
\centering
\includegraphics[width=3.5in]{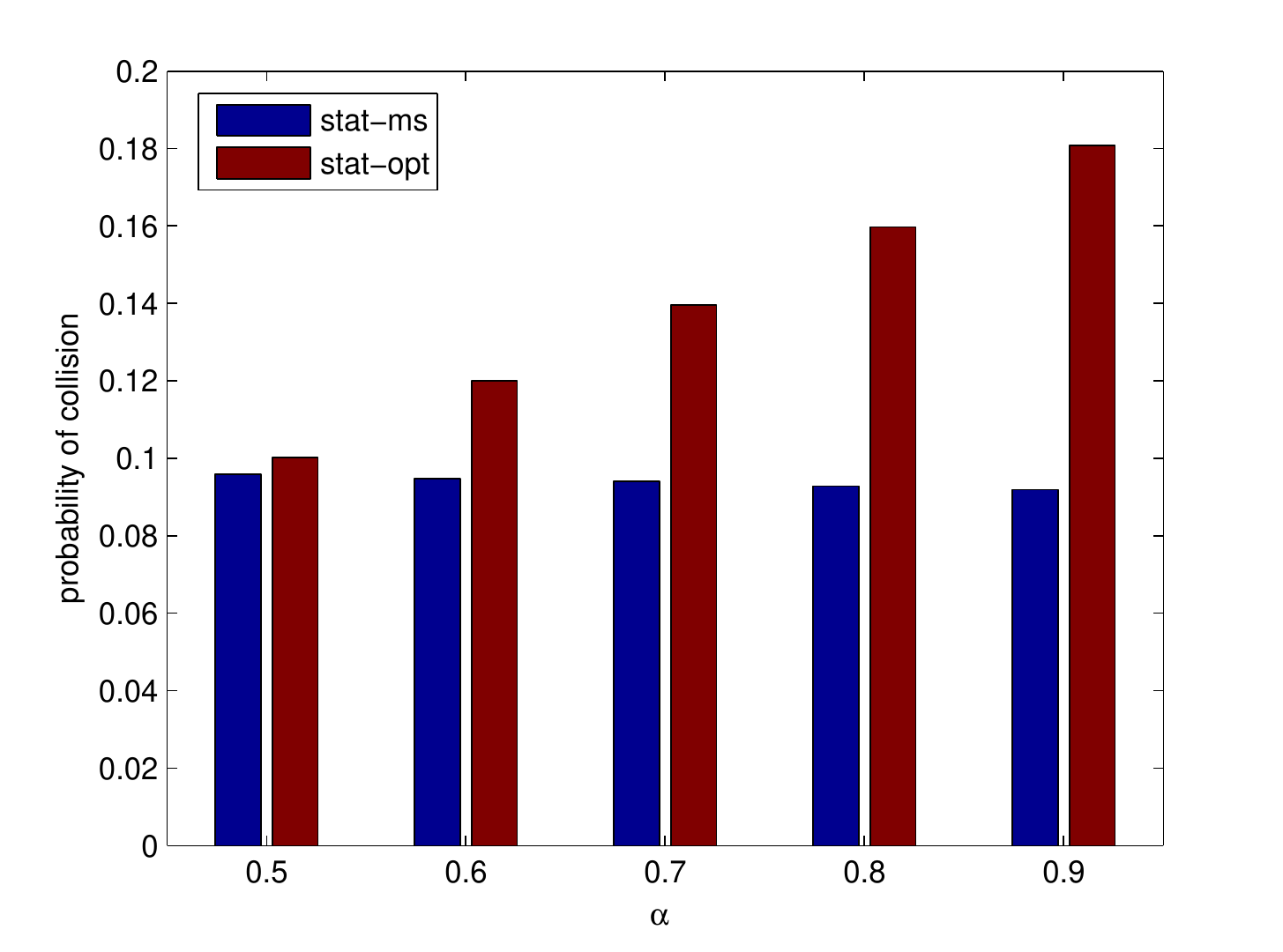}
\caption{Probability of collision of statistical PTSI optimal and multiple-shot strategies when the SU has incorrect information of $\alpha$.}
\label{fig:pcbar}
\end{figure}

\section{Experimental Study}\label{sec:6}
This section provides an experimental study of the primary traffic and the SU performance using a software-defined radio (SDR) device working on the 2.4 GHz frequency band. The Wi-Fi users are the primary users. Firstly, the traffic of the primary network is studied, especially its non-stationarity. The performances of the SU strategies are also studied. Here we are mainly interested in the performance of statistical PTSI optimal and multiple-shot strategies under the non-stationary primary traffic.

\subsection{Primary traffic non-stationarity}
The channel state is observed using the SDR. Samples of channel idle time durations are extracted from the observations. To study the non-stationarity of the channel idle time, the sample set is divided into small groups in a sequential order. The parameters $\alpha_i$'s and $\lambda_i$'s are estimated for each of the small groups using $N=2$. We study three data sets observed in different times of day. There are $100,000$ channel idle time samples in each data set. Furthermore, each set is divided into groups with $1000$ samples. The box plot of the estimates is shown in Figure \ref{fig:nonstationarity}. One can make two key observations from the data. Firstly, the $\alpha_i$'s and $\lambda_i$'s are not static within a data set, which means that there is variation within short periods of times. Secondly, the estimates from different data sets are quite distinct. Therefore, in practical situations, the non-stationarity does exist in primary traffic and it poses a huge challenge for the SU.
\begin{figure*}[!t]
  \centering
  \subfloat[$\lambda$ estimates]{\label{fig:lambdadata}\includegraphics[width=3.5in]{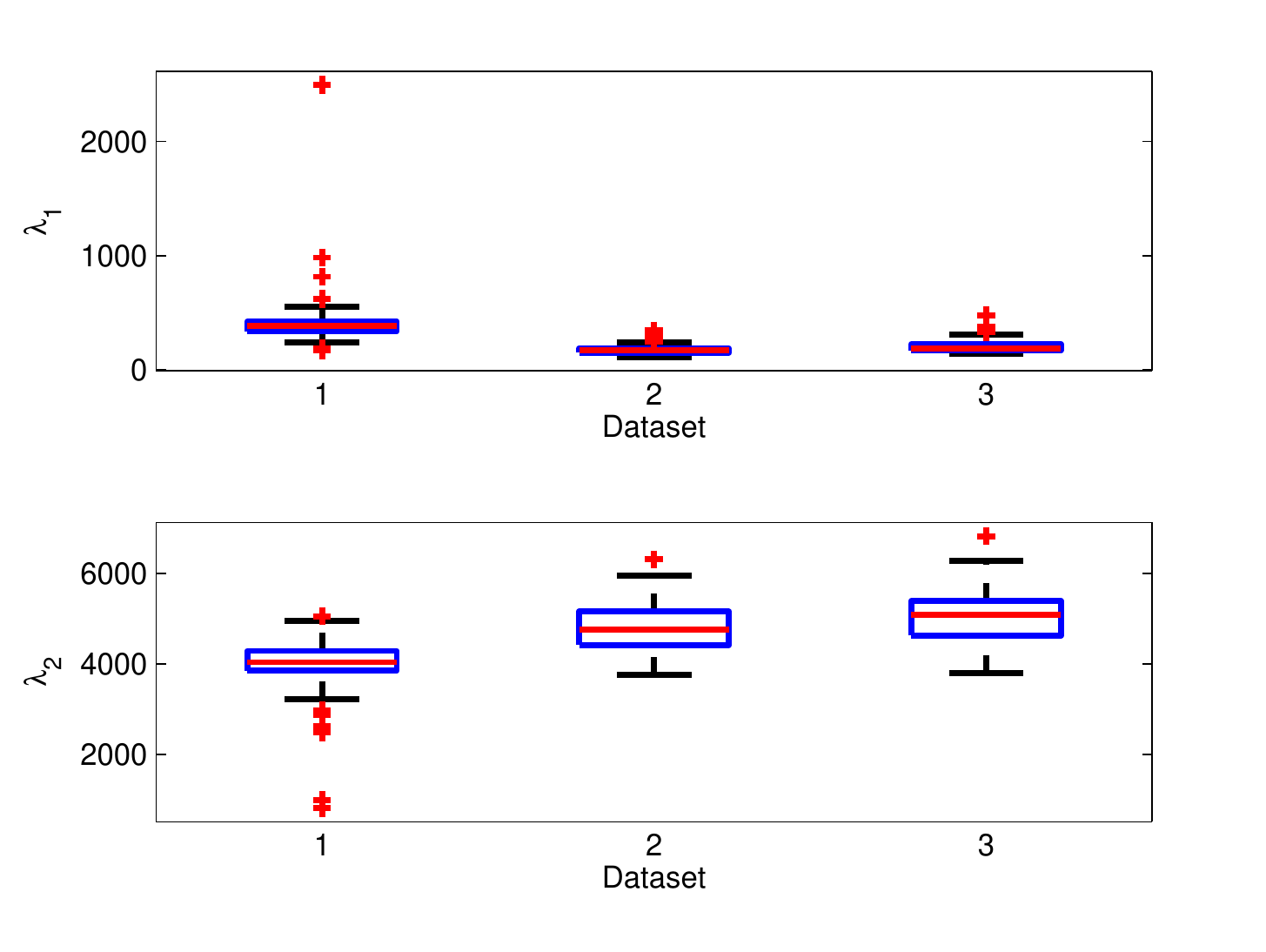}}
  \subfloat[$\alpha$ estimates]{\label{fig:alphadata}\includegraphics[width=3.5in]{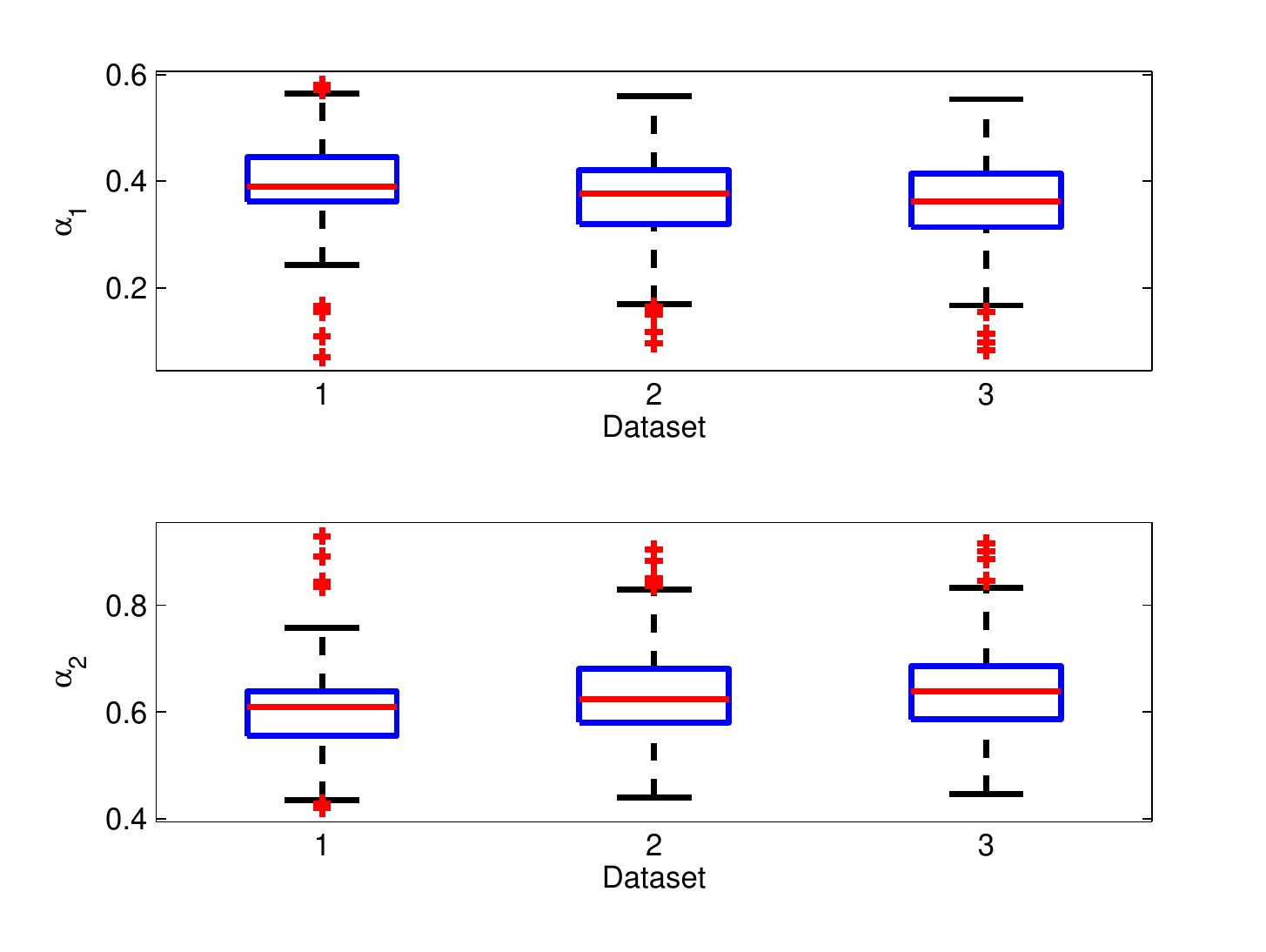}}
  \caption{Box plot of statistics of the parameter estimates in hyper-exponential distribution.}
  \label{fig:nonstationarity}
\end{figure*}

\subsection{SU performance in actual primary traffic}
According to experimental results in \cite{kundargi2010protomac}, the primary network (in our case is the WiFi network) has a hard limit on the amounts of interference that it can tolerate. Therefore it makes sense to use the probability of collision to describe the amount of interference. Although in the strategies we have discussed, the probability of collision constraint is a long-term average, it is also desirable that the probability of collision in short periods of time be bounded by the constraint. We will call the probability of collision in short periods of time as the \emph{instantaneous} probability of collision. Another metric that is used to describe the performance of the strategy is the probability that the instantaneous probability of collision exceeds the constraint. We will call this metric as the \emph{outage probability}.

\par
In the experiment, the transmission strategies which correspond to an estimate of the hyper-exponential distribution and a desired probability of collision are stored in the SU. The optimal and multiple-shot strategies are then applied to different time periods. An initial estimate of the hyper-exponential distribution is $\lambda^0=[160, 3670]$, $\alpha^0=[0.32, 0.68]$. Three data sets with $100,000$ channel idle time samples each are studied. Three performance metrics, the overall channel capacity, the overall probability of collision and the outage probability are used to compare between the statistical PTSI optimal and multiple-shot strategies. The desired probability of collision is $\eta=0.05$. Figure \ref{fig:oscandopc} plots the overall secondary capacity and probability of collision. In data sets 2 and 3, the secondary capacities and probability of collisions of the optimal strategy are substantially higher than the multiple shot strategy, but lower in data set 1. The overall probability of collision is close to the designed value. To explain this phenomenon, we produce the estimates using data set 1, 2, and 3: $\lambda^1=[349, 4005]$, $\alpha^1=[0.38, 0.62]$; $\lambda^2=[163, 4688]$, $\alpha^2=[0.35, 0.65]$; $\lambda^3=[186, 4961]$, $\alpha^3=[0.36, 0.64]$. We observe that the data sets 2 and 3 are more ``similar'' to the previous estimates, especially for the smaller $\lambda$. This indicates that if the estimates of the corresponding strategy is close to the actual value, the optimal strategy performs well. Otherwise, its performance drops significantly, as shown in data set 1. However, the multiple-shot strategy performs quite consistently for the three data sets. A similar conclusion can be made by inspecting the outage probability in Figure \ref{fig:outages}. Figure \ref{fig:outage} shows that the multiple-shot strategy has much lower outage probability than the optimal strategy. The box plot of the probabilities of collisions in Figure \ref{fig:outagebox} gives a more detailed description of the distribution of the probabilities of collisions. The black horizontal line is the designed probability of collision. One can see that the probabilities of collisions of the optimal strategy has a much wider span and statistically higher value than that of the multiple-shot strategy.
\begin{figure*}[!t]
  \centering
  \subfloat[Overall secondary capacity.]{\label{fig:osc}\includegraphics[width=3.5in]{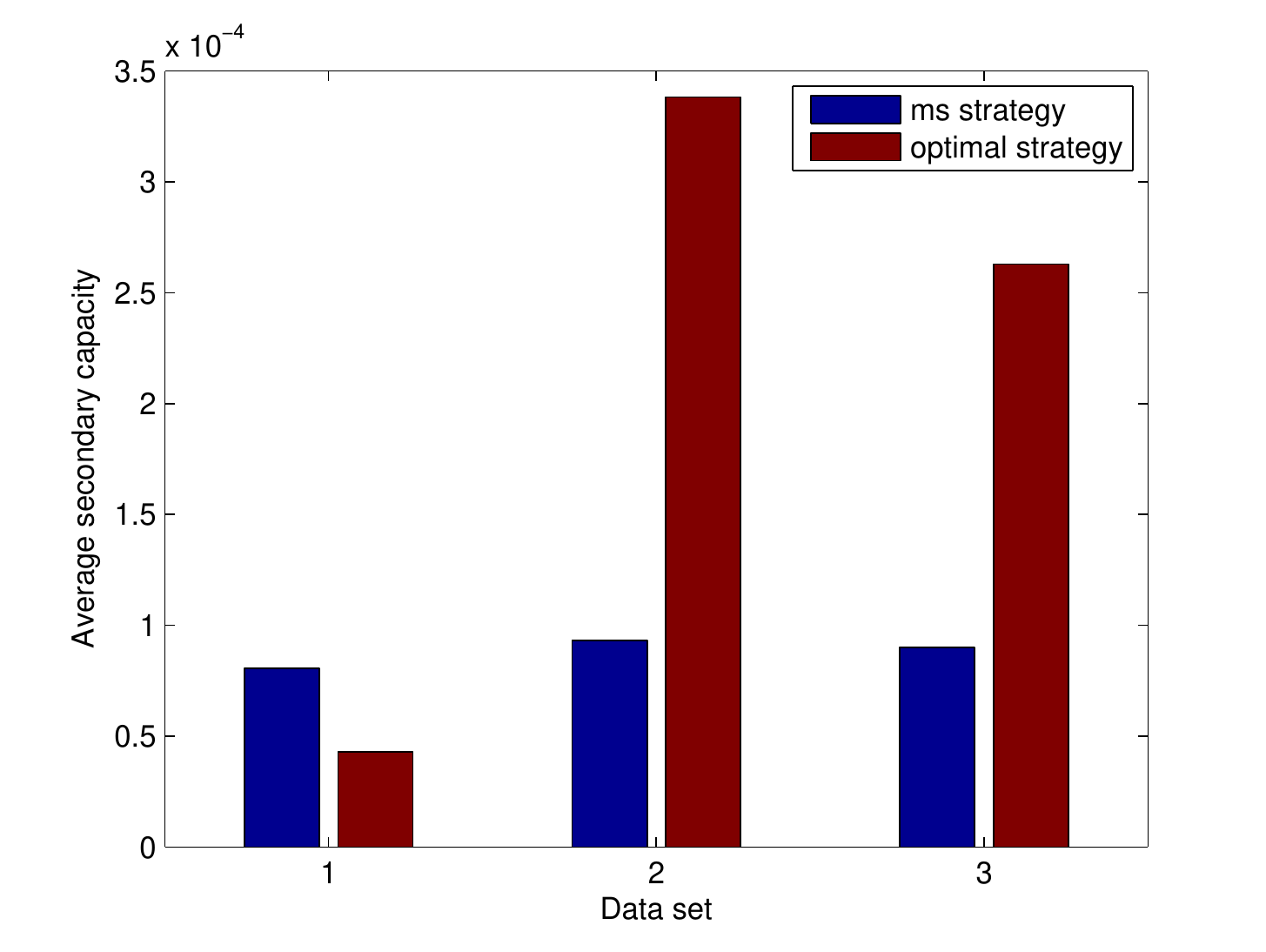}}
  \subfloat[Overall probability of collision.]{\label{fig:opc}\includegraphics[width=3.5in]{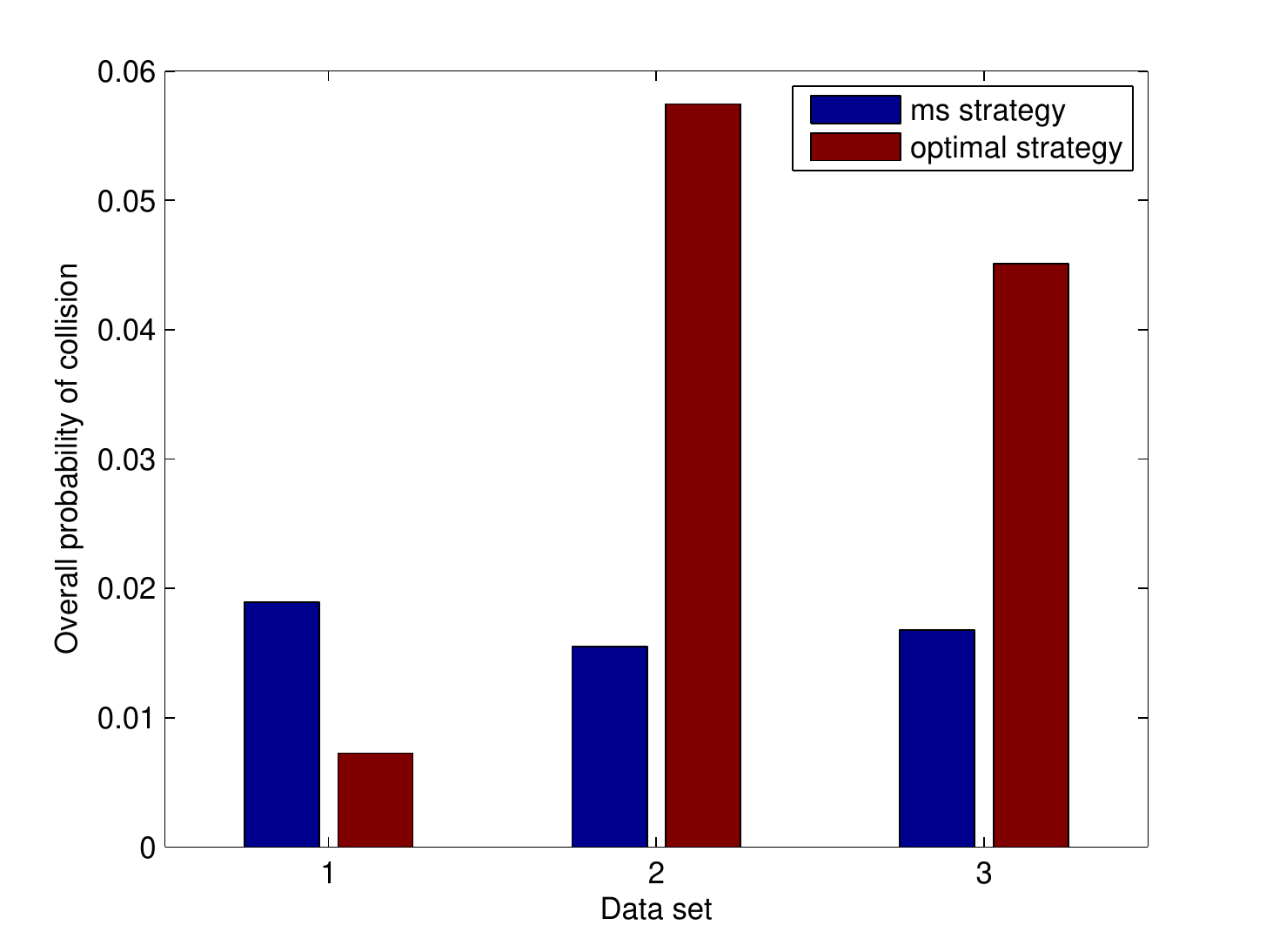}}
  \caption{Overall secondary capacity and probability of collision in three data sets.}
  \label{fig:oscandopc}
\end{figure*}
\begin{figure*}[!t]
  \centering
  \subfloat[Outage]{\label{fig:outage}\includegraphics[width=3.5in]{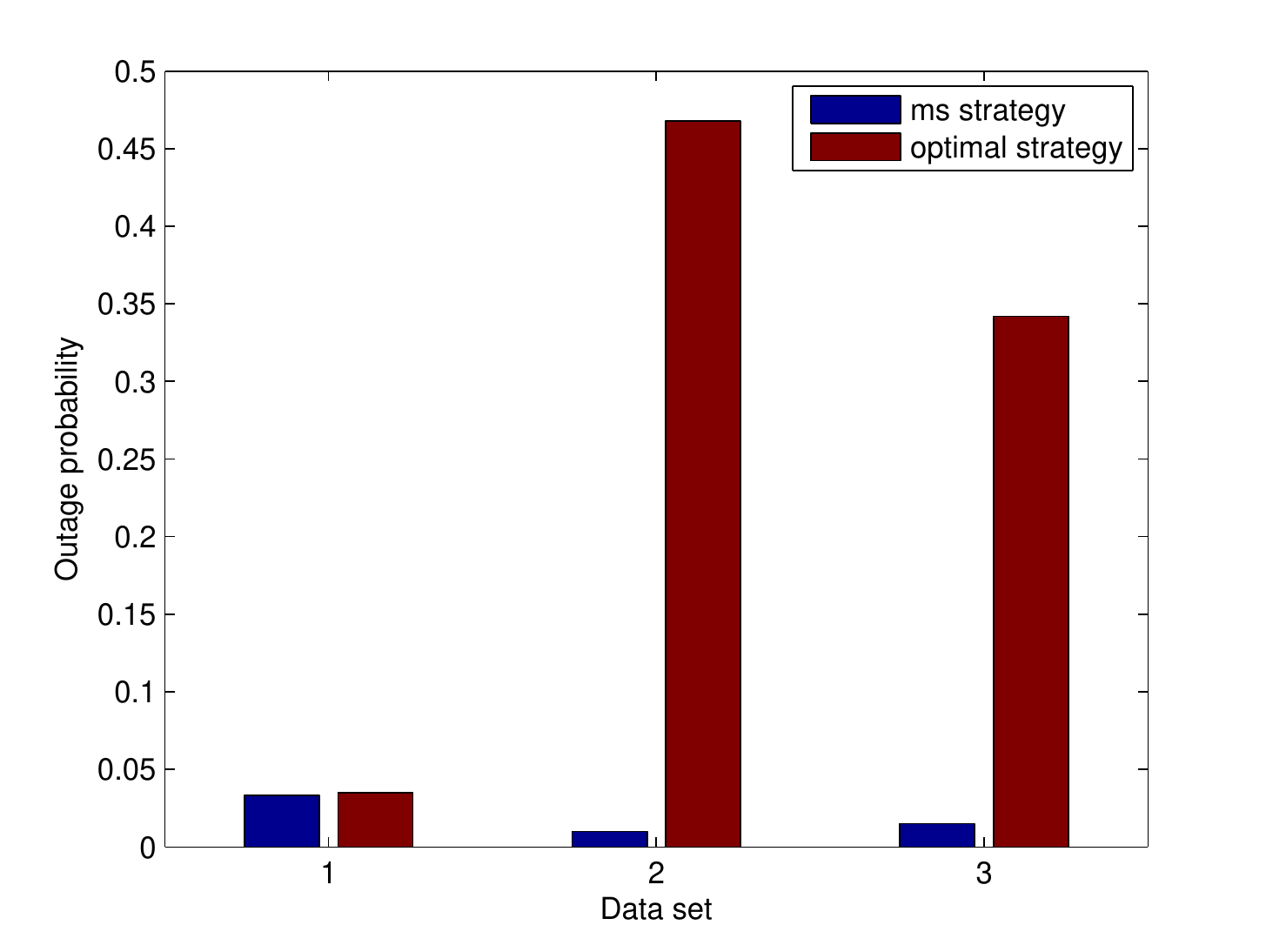}}
  \subfloat[Probability of collision boxplot]{\label{fig:outagebox}\includegraphics[width=3.5in]{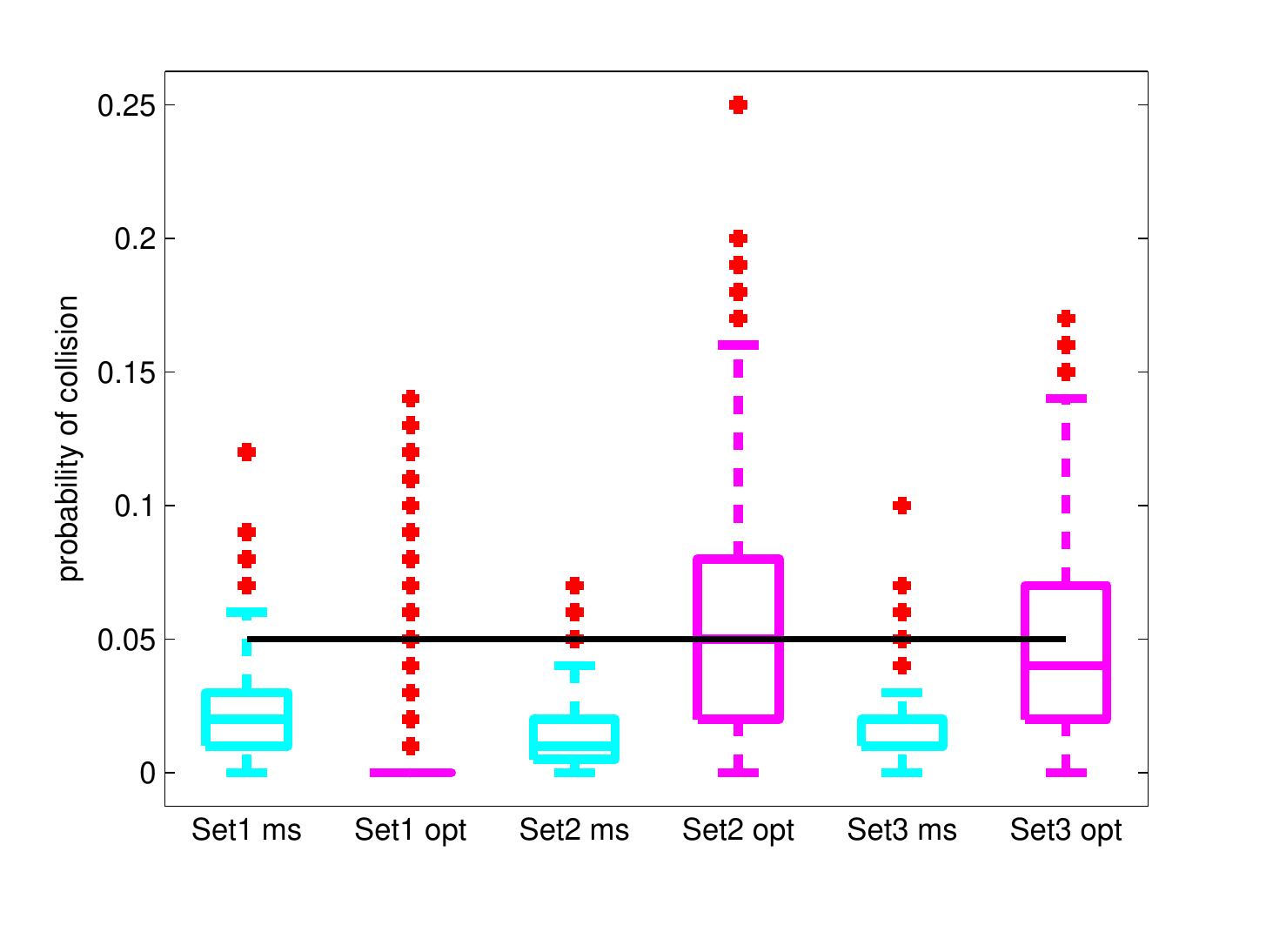}}
  \caption{Outage probability evaluated every 100 channel idle times.}
  \label{fig:outages}
\end{figure*}

\par
To sum up, the optimal strategy achieves the highest secondary capacity at the expense of a much higher outage probability in the presence of non-stationarity. Also, it suffers from a large performance loss when primary traffic model is incorrect. The multiple-shot strategy achieves a suboptimal secondary capacity. However, its secondary capacity is more stable with non-stationary primary traffic. Additionally, it maintains a very low outage probability.

\section{Conclusion}\label{sec:7}
In this work, we first proposed a hyper-exponential distribution to fit in the CITD data observed from real-world WLAN networks. The hyper-exponential distribution provides not only a good fit to the data, but also a well-structured primary traffic model. Based on this distribution, a SMMPP model for the primary traffic pattern is proposed. A novel concept of primary traffic state information (PSTI) is derived to describe the amount of information the SU has about the primary traffic. Three types of PTSI and their practicality are discussed along with the associated secondary transmission strategies, including the statistical PTSI, Markov-level PTSI and full PTSI.

\par
Currently, the most practical approach is to use the statistical PTSI. The optimal strategy with this PTSI is developed. We show that it achieves the highest secondary capacity when the primary traffic is stationary. But in non-stationary primary traffic, this strategy constantly produces higher probability of collision exceeding the designed value. This of course is not a desirable situation for the primary users. To this end, a multiple-shot strategy is proposed, which only achieves a secondary capacity similar to the full-PTSI balanced strategy in stationary primary traffic but has a constant probability of collision performance even if the primary traffic is non-stationary.

\appendices
\section{Proofs}\label{app2}
\subsection{Proof of Proposition \ref{sec3prop1}}
\begin{IEEEproof}
Firstly, we prove that the optimal solution takes on the form
\begin{align}
    x_i^{\ast}(t)=
    \begin{cases}
    1,\ \frac{1-F_i(t)}{f_i(t)}>a_i^{\ast},\\
    p^{\ast},\ \frac{1-F_i(t)}{f_i(t)}=a_i^{\ast},\\
    0,\ \frac{1-F_i(t)}{f_i(t)}<a_i^{\ast},\nonumber
    \end{cases}
\end{align}
for all $i$. This result is a straightforward extension of the result in \cite{huang2009optimal}. Then we prove that for the index set $I\subseteq\{1, 2, \ldots, N\}$, $a_i^{\ast}=a^{\ast}$, $\forall i\in I$ and $\min\limits_{t\geq 0}\frac{1-F_i(t)}{f_i(t)}>a^{\ast}$ when $i\in\{1, 2, \ldots, N\}\setminus I$. Assume without loss of generality that $a_k^{\ast}>a_l^{\ast}$, $\int_{t:\frac{1-F_k(t)}{f_k(t)}>a_k^{\ast}}x_k^{\ast}(t)f_k(t)dt+\int_{t:\frac{1-F_l(t)}{f_l(t)}>a_l^{\ast}}x_l^{\ast}(t)f_l(t)dt=\eta$. Then there must exist a $b$ such that $\int_{t:\frac{1-F_k(t)}{f_k(t)}>b}x_k(t)f_k(t)dt+\int_{t:\frac{1-F_l(t)}{f_l(t)}>b}x_l(t)f_i(t)dt=\eta$, $a_k^{\ast}>b>a_l^{\ast}$. Then one can verify \eqref{app2_eq1}. The last equality of \eqref{app2_eq1} comes from the fact that
\begin{align*}
    &\int_{t:\frac{1-F_k(t)}{f_k(t)}>a_k^{\ast}}x_k^{\ast}(t)f_k(t)dt+\int_{t:\frac{1-F_l(t)}{f_l(t)}>a_l^{\ast}}x_l^{\ast}(t)f_l(t)dt\nonumber\\
    &=\int_{t:\frac{1-F_k(t)}{f_k(t)}>b}x_k(t)f_k(t)dt+\int_{t:\frac{1-F_l(t)}{f_l(t)}>b}x_l(t)f_i(t)dt\nonumber\\
    \Leftrightarrow &\int_{t:\frac{1-F_k(t)}{f_k(t)}>b}x_k(t)f_k(t)dt-\int_{t:\frac{1-F_k(t)}{f_k(t)}>a_k^{\ast}}x_k^{\ast}(t)f_k(t)dt\nonumber\\
    &=\int_{t:\frac{1-F_l(t)}{f_l(t)}>a_l^{\ast}}x_l^{\ast}(t)f_l(t)dt-\int_{t:\frac{1-F_l(t)}{f_l(t)}>b}x_l(t)f_i(t)dt\nonumber
\end{align*}
Hence we arrive at a contradiction. Then the result of the proposition follows.
\end{IEEEproof}
\begin{figure*}[!t]
\normalsize
\begin{align}
\label{app2_eq1}
&\int_{t:\frac{1-F_k(t)}{f_k(t)}\geq b}x_k(t)(1-F_k(t))dt+\int_{t:\frac{1-F_l(t)}{f_l(t)}\geq b}x_l(t)(1-F_l(t))dt\nonumber\\
    &-\int_{t:\frac{1-F_k(t)}{f_k(t)}\geq a_k^{\ast}}x_k^{\ast}(t)(1-F_k(t))dt-\int_{t:\frac{1-F_l(t)}{f_l(t)}\geq a_l^{\ast}}x_l^{\ast}(t)(1-F_l(t))dt\nonumber\\
    =&\int_{t:b\leq\frac{1-F_k(t)}{f_k(t)}<a_k^{\ast}}x_k(t)(1-F_k(t))dt-\int_{t:a_l^{\ast}\leq\frac{1-F_l(t)}{f_l(t)}<b}x_l^{\ast}(t)(1-F_l(t))dt\nonumber\\
    =&\int_{t:b\leq\frac{1-F_k(t)}{f_k(t)}<a_k^{\ast}}x_k(t)f_k(t)\frac{1-F_k(t)}{f_k(t)}dt-\int_{t:a_l^{\ast}\leq\frac{1-F_l(t)}{f_l(t)}<b}x_l^{\ast}(t)f_l(t)\frac{1-F_l(t)}{f_l(t)}dt\nonumber\\
    \geq &b\int_{t:b\leq\frac{1-F_k(t)}{f_k(t)}<a_k^{\ast}}x_k(t)f_k(t)dt-b\int_{t:a_l^{\ast}\leq\frac{1-F_l(t)}{f_l(t)}<b}x_l^{\ast}(t)f_l(t)dt\nonumber\\
    =&0.
\end{align}
\hrulefill
\vspace*{4pt}
\end{figure*}

\subsection{Proof of Proposition \ref{sec4prop1}}
\begin{proof}
The probability of collision under the multiple-shot strategy can be written as
\begin{align}
    &P(\bigcup_{i=1}^{N}\{t_{i+1}^e\leq X<t_{i+1}^e+t_i\})\nonumber\\
    &=\sum_{j=1}^{N}\alpha_jP(\bigcup_{i=1}^{N}\{t_{i+1}^e\leq X<t_{i+1}^e+t_i\}|\lambda_j),\nonumber
\end{align}
where $t_{N+1}^e=0$. Each term inside the summation can be expanded as
\begin{align}\label{sec4eq1}
    &P(\bigcup_{i=1}^{N}\{t_{i+1}^e\leq X<t_{i+1}^e+t_i\}|\lambda_j)\nonumber\\
    =&P(\bigcup_{i=1}^{j-1}\{t_{i+1}^e\leq X<t_{i+1}^e+t_i\}|\lambda_j)\nonumber\\
    &+P(\bigcup_{i=j+1}^{N}\{t_{i+1}^e\leq X<t_{i+1}^e+t_i\}|\lambda_j)\nonumber\\
    &+P(t_{j+1}^e\leq X<t_{j+1}^e+t_j|\lambda_j).
\end{align}
Since $1-e^{-\lambda_{j}t_{j}^e}= 1-\epsilon$, the first term in \eqref{sec4eq1} is bounded by $\epsilon$
\begin{align}
    &P(\bigcup_{i=1}^{j-1}\{t_{i+1}^e\leq X<t_{i+1}^e+t_i\}|\lambda_j)\nonumber\\
    < & P(X>t_j^e|\lambda_j)\nonumber\\
    < & \epsilon.\nonumber
\end{align}
Then consider the following element in the second term
\begin{align}
    &P(t_{i+1}^e\leq X<t_{i+1}^e+t_i|\lambda_j)\nonumber\\
    =&e^{-\lambda_jt_{i+1}^e}(1-e^{-\lambda_jt_i})\nonumber\\
    =&\epsilon^{\frac{\lambda_j}{\lambda_{i+1}}}(1-e^{-\frac{\lambda_j}{\lambda_i}\ln\frac{1}{1-\eta}})\nonumber\\
    \leq & (1-e^{-\frac{\lambda_j}{\lambda_i}\ln\frac{1}{1-\eta}})\nonumber\\
    \approx & 0.\nonumber
\end{align}
The equality is obtained by substituting $t_{i+1}^e=\frac{1}{\lambda_{i+1}}\ln\frac{1}{\epsilon}$, and $t_i=\frac{1}{\lambda_{i}}\ln\frac{1}{1-\eta}$. The approximations are obtained according to the assumption that $\lambda_{i+1}\gg \lambda_j$ and $\lambda_i\gg \lambda_j$. Then \eqref{sec4eq1} can be approximated by
\begin{align}
    &P(\bigcup_{i=1}^{N}\{t_{i+1}^e\leq X<t_{i+1}^e+t_i\})\nonumber\\
    \approx & \sum_{j=1}^{N}\alpha_jP(t_{j+1}^e\leq X<t_{j+1}^e+t_j|\lambda_j)\nonumber\\
    =&\sum_{j=1}^{N}\alpha_je^{-\lambda_jt_{j+1}^e}(1-e^{-\lambda_jt_j})\nonumber\\
    \leq & \sum_{j=1}^{N}\alpha_j\eta\nonumber\\
    =&\eta.\nonumber
\end{align}
Hence prove the proposition.
\end{proof}

\ifCLASSOPTIONcaptionsoff
  \newpage
\fi

\bibliographystyle{IEEEtran}
\bibliography{IEEEabrv,bibfile}

\begin{thebibliography}{10}
\providecommand{\url}[1]{#1}
\csname url@samestyle\endcsname
\providecommand{\newblock}{\relax}
\providecommand{\bibinfo}[2]{#2}
\providecommand{\BIBentrySTDinterwordspacing}{\spaceskip=0pt\relax}
\providecommand{\BIBentryALTinterwordstretchfactor}{4}
\providecommand{\BIBentryALTinterwordspacing}{\spaceskip=\fontdimen2\font plus
\BIBentryALTinterwordstretchfactor\fontdimen3\font minus
  \fontdimen4\font\relax}
\providecommand{\BIBforeignlanguage}[2]{{%
\expandafter\ifx\csname l@#1\endcsname\relax
\typeout{** WARNING: IEEEtran.bst: No hyphenation pattern has been}%
\typeout{** loaded for the language `#1'. Using the pattern for}%
\typeout{** the default language instead.}%
\else
\language=\csname l@#1\endcsname
\fi
#2}}
\providecommand{\BIBdecl}{\relax}
\BIBdecl

\bibitem{mitola1999cognitive}
J.~Mitola~III and G.~Maguire~Jr, ``Cognitive radio: making software radios more
  personal,'' \emph{Personal Communications, IEEE}, vol.~6, no.~4, pp. 13--18,
  1999.

\bibitem{ji2007cognitive}
Z.~Ji and K.~Liu, ``Cognitive radios for dynamic spectrum access-dynamic
  spectrum sharing: A game theoretical overview,'' \emph{Communications
  Magazine, IEEE}, vol.~45, no.~5, pp. 88--94, 2007.

\bibitem{win2009mathematical}
M.~Win, P.~Pinto, and L.~Shepp, ``A mathematical theory of network interference
  and its applications,'' \emph{Proceedings of the IEEE}, vol.~97, no.~2, pp.
  205--230, 2009.

\bibitem{huang2009optimal}
S.~Huang, X.~Liu, and Z.~Ding, ``Optimal transmission strategies for dynamic
  spectrum access in cognitive radio networks,'' \emph{IEEE Transactions on
  Mobile Computing}, pp. 1636--1648, 2009.

\bibitem{geirhofer2008cognitive}
S.~Geirhofer, L.~Tong, and B.~Sadler, ``Cognitive medium access: constraining
  interference based on experimental models,'' \emph{Selected Areas in
  Communications, IEEE Journal on}, vol.~26, no.~1, pp. 95--105, 2008.

\bibitem{zhao2007decentralized}
Q.~Zhao, L.~Tong, A.~Swami, and Y.~Chen, ``Decentralized cognitive mac for
  opportunistic spectrum access in ad hoc networks: A pomdp framework,''
  \emph{Selected Areas in Communications, IEEE Journal on}, vol.~25, no.~3, pp.
  589--600, 2007.

\bibitem{lai2010cognitive}
L.~Lai, H.~El~Gamal, H.~Jiang, and H.~Poor, ``Cognitive medium access:
  Exploration, exploitation, and competition,'' \emph{IEEE Transactions on
  Mobile Computing}, pp. 239--253, 2010.

\bibitem{4907431}
C.-X. Wang, X.~Hong, H.-H. Chen, and J.~Thompson, ``On capacity of cognitive
  radio networks with average interference power constraints,'' \emph{Wireless
  Communications, IEEE Transactions on}, vol.~8, no.~4, pp. 1620--1625, 2009.

\bibitem{4373439}
W.~Wu, S.~Vishwanath, and A.~Arapostathis, ``Capacity of a class of cognitive
  radio channels: Interference channels with degraded message sets,''
  \emph{Information Theory, IEEE Transactions on}, vol.~53, no.~11, pp.
  4391--4399, 2007.

\bibitem{4786456}
X.~Kang, Y.-C. Liang, A.~Nallanathan, H.~Garg, and R.~Zhang, ``Optimal power
  allocation for fading channels in cognitive radio networks: Ergodic capacity
  and outage capacity,'' \emph{Wireless Communications, IEEE Transactions on},
  vol.~8, no.~2, pp. 940--950, 2009.

\bibitem{5419086}
H.~Suraweera, P.~Smith, and M.~Shafi, ``Capacity limits and performance
  analysis of cognitive radio with imperfect channel knowledge,''
  \emph{Vehicular Technology, IEEE Transactions on}, vol.~59, no.~4, pp.
  1811--1822, 2010.

\bibitem{4155368}
S.~Jafar and S.~Srinivasa, ``Capacity limits of cognitive radio with
  distributed and dynamic spectral activity,'' \emph{Selected Areas in
  Communications, IEEE Journal on}, vol.~25, no.~3, pp. 529--537, 2007.

\bibitem{5208469}
A.~Jovicic and P.~Viswanath, ``Cognitive radio: An information-theoretic
  perspective,'' \emph{Information Theory, IEEE Transactions on}, vol.~55,
  no.~9, pp. 3945--3958, 2009.

\bibitem{crovella1998heavy}
M.~Crovella, M.~Taqqu, and A.~Bestavros, ``Heavy-tailed probability
  distributions in the world wide web,'' \emph{A practical guide to heavy
  tails: statistical techniques and applications}, pp. 3--25, 1998.

\bibitem{paxson1995wide}
V.~Paxson and S.~Floyd, ``Wide area traffic: the failure of poisson modeling,''
  \emph{IEEE/ACM Transactions on Networking (ToN)}, vol.~3, no.~3, pp.
  226--244, 1995.

\bibitem{karagiannis2010power}
T.~Karagiannis, J.~Le~Boudec, and M.~Vojnovi{\'c}, ``Power law and exponential
  decay of intercontact times between mobile devices,'' \emph{IEEE Transactions
  on Mobile Computing}, pp. 1377--1390, 2010.

\bibitem{heffes1986markov}
H.~Heffes and D.~Lucantoni, ``A markov modulated characterization of packetized
  voice and data traffic and related statistical multiplexer performance,''
  \emph{Selected Areas in Communications, IEEE Journal on}, vol.~4, no.~6, pp.
  856--868, 1986.

\bibitem{alliance2010wi}
W.-F. Alliance, ``Wi-fi certified wi-fi direct,'' \emph{Personal, Portable
  Wi-Fi Technology}, 2010.

\bibitem{feldmann1998fitting}
A.~Feldmann and W.~Whitt, ``Fitting mixtures of exponentials to long-tail
  distributions to analyze network performance models* 1,'' \emph{Performance
  evaluation}, vol.~31, no. 3-4, pp. 245--279, 1998.

\bibitem{liu2010novel}
Y.~Liu, N.~Kundargi, and A.~Tewfik, ``A novel sense-transmit-wait strategy for
  coexistence of cognitive radio networks with ieee 802.11 wlans,'' in
  \emph{Communications, Control and Signal Processing (ISCCSP), 2010 4th
  International Symposium on}.\hskip 1em plus 0.5em minus 0.4em\relax IEEE, pp.
  1--5.

\bibitem{liuhyperexponential}
Y.~Liu and A.~Tewfik, ``Hyperexponential approximation of channel idle time
  distribution with implication to secondary transmission strategy,'' in
  \emph{Communications (ICC), 2012 IEEE International Conference on}.\hskip 1em
  plus 0.5em minus 0.4em\relax IEEE, 2010, pp. 1--5.

\bibitem{kershaw2007kismet}
``Gnu radio - usrp2, gnuradio.org/redmine/projects/gnuradio/wiki/usrp2,''
  \emph{Retrieved from the Web}.

\bibitem{feller2008introduction}
W.~Feller, \emph{An introduction to probability theory and its
  applications}.\hskip 1em plus 0.5em minus 0.4em\relax Wiley-India, 2008,
  vol.~2.

\bibitem{bilmes1998gentle}
J.~Bilmes, ``A gentle tutorial of the em algorithm and its application to
  parameter estimation for gaussian mixture and hidden markov models,''
  \emph{International Computer Science Institute}, vol.~4, p. 126, 1998.

\bibitem{leland1994self}
W.~Leland, M.~Taqqu, W.~Willinger, and D.~Wilson, ``On the self-similar nature
  of ethernet traffic (extended version),'' \emph{Networking, IEEE/ACM
  Transactions on}, vol.~2, no.~1, pp. 1--15, 1994.

\bibitem{duarte2010full}
M.~Duarte and A.~Sabharwal, ``Full-duplex wireless communications using
  off-the-shelf radios: Feasibility and first results,'' in \emph{Signals,
  Systems and Computers (ASILOMAR), 2010 Conference Record of the Forty Fourth
  Asilomar Conference on}.\hskip 1em plus 0.5em minus 0.4em\relax IEEE, 2010,
  pp. 1558--1562.

\bibitem{ahmed2012simultaneous}
E.~Ahmed, A.~Eltawil, and A.~Sabharwal, ``Simultaneous transmit and sense for
  cognitive radios using full-duplex: A first study,'' in \emph{Antennas and
  Propagation Society International Symposium (APSURSI), 2012 IEEE}.\hskip 1em
  plus 0.5em minus 0.4em\relax IEEE, 2012, pp. 1--2.

\bibitem{duarte2012design}
M.~Duarte, A.~Sabharwal, V.~Aggarwal, R.~Jana, K.~Ramakrishnan, C.~Rice, and
  N.~Shankaranarayanan, ``Design and characterization of a full-duplex
  multi-antenna system for wifi networks,'' \emph{arXiv preprint
  arXiv:1210.1639}, 2012.

\bibitem{choi2010achieving}
J.~I. Choi, M.~Jain, K.~Srinivasan, P.~Levis, and S.~Katti, ``Achieving single
  channel, full duplex wireless communication,'' in \emph{Proceedings of the
  sixteenth annual international conference on Mobile computing and
  networking}.\hskip 1em plus 0.5em minus 0.4em\relax ACM, 2010, pp. 1--12.

\bibitem{peskir2002solving}
G.~Peskir and A.~Shiryaev, ``Solving the poisson disorder problem,''
  \emph{Advances in finance and stochastics}, pp. 295--312, 2002.

\bibitem{kundargi2010protomac}
N.~Kundargi and A.~Tewfik, ``Protomac: Proactive transmit opportunity detection
  at the mac layer for cognitive radios,'' in \emph{Communications (ICC), 2010
  IEEE International Conference on}.\hskip 1em plus 0.5em minus 0.4em\relax
  IEEE, 2010, pp. 1--5.

\end{thebibliography}

\end{document}